\documentclass[10pt,twoside,english]{article}
\usepackage[T1]{fontenc}
\usepackage[latin9]{inputenc}
\usepackage{geometry}
\usepackage{color}
\usepackage{babel}
\usepackage{amsthm}
\usepackage{amsmath}
\usepackage{amssymb}
\usepackage{graphicx}
\usepackage{esint}
\usepackage[unicode=true,pdfusetitle,
 bookmarks=true,bookmarksnumbered=false,bookmarksopen=false,
 breaklinks=false,pdfborder={0 0 1},backref=false,colorlinks=false]
 {hyperref}
\usepackage{breakurl}

\makeatletter

\providecommand{\tabularnewline}{\\}

\numberwithin{equation}{section}
\numberwithin{figure}{section}
\theoremstyle{plain}
\newtheorem{thm}{\protect\theoremname}
  \theoremstyle{remark}
  \newtheorem*{rem*}{\protect\remarkname}

\date{}

\@ifundefined{showcaptionsetup}{}{%
 \PassOptionsToPackage{caption=false}{subfig}}
\usepackage{subfig}
\makeatother

  \providecommand{\remarkname}{Remark}
\providecommand{\theoremname}{Theorem}

\begin{document}

\title{Study of a model equation in detonation theory}

\author{Luiz M. Faria$^{*}$, Aslan R. Kasimov%
\thanks{Applied Mathematics and Computational Sciences, KAUST, Saudi Arabia
\protect \\
luiz.faria@kaust.edu.sa, aslan.kasimov@kaust.edu.sa%
} , and Rodolfo R. Rosales%
\thanks{Department of Mathematics, MIT, USA, rrr@math.mit.edu%
}}
\maketitle
\begin{abstract}
Here we analyze properties of an equation that we previously proposed
to model the dynamics of unstable detonation waves \cite{kasimov2013model}.
The equation is 
\[
u_{t}+\frac{1}{2}\left(u^{2}-uu\left(0_{-},t\right)\right)_{x}=f\left(x,u\left(0_{-},t\right)\right),\quad x\le0,\quad t>0.
\]
It describes a detonation shock at $x=0$ with the reaction zone in
$x<0$. We investigate the nature of the steady-state solutions of
this nonlocal hyperbolic balance law, the linear stability of these
solutions, and the nonlinear dynamics. We establish the existence
of instability followed by a cascade of period-doubling bifurcations
leading to chaos. 
\end{abstract}

\section{Introduction}

A detonation is a shock wave that propagates in a reactive medium
where exothermic chemical reactions are ignited as a result of the
heating by the shock compression. The energy released in these reactions,
in turn, feeds back to the shock in the form of compression waves
and thus sustains the shock motion. The dynamics of such shock--reaction
coupling is highly nonlinear due to the sensitivity of the chemical
reactions to temperature, making the problem significantly more challenging
than shock dynamics in non-reactive media. A steady planar detonation
wave is rarely observed in experiments. Complex time-dependent and
multi-dimensional structures tend to develop \cite{fickett2011detonation,Lee-2008}.
Numerical simulations of the equations of reactive gas dynamics are
able to reproduce at a qualitative level the complex structures observed
in experiments (see, e.g., \cite{taki1978numerical,BM92,oran1998numerical}).
However, obtaining physical insights into the basic mechanisms of
the instability requires simplified modeling and remains challenging.

In one dimension, the instabilities of the reactive shock wave manifest
themselves in the form of a ``galloping detonation'' \cite{fickett1966flow,fickett2011detonation},
wherein the shock speed oscillates around its steady value. It has
been shown through extensive numerical experiments that as the activation
energy, $E$, a parameter in the equations measuring the temperature
sensitivity of the chemical reactions, is varied, the shock speed
transitions from a constant to an oscillatory function. Further increase
of $E$ leads to a period-doubling bifurcation cascade, which ultimately
results in the shock moving at a chaotic speed \cite{Ng2005,HenrickAslamPowers2006}.
The mechanism for such instabilities is still not completely understood. 

In this paper, we show that the model introduced in \cite{kasimov2013model},
which consists of a single non-local partial differential equation
(PDE), is capable of reproducing the complexity observed in one-dimensional
simulations of reactive Euler equations. The model possesses traveling
wave solutions precisely analogous to the ZND theory (named after
Zel'dovich \cite{Zeldovich1940}, von Neumann \cite{vonNeumann1942},
and D\"{o}ring \cite{Doering1943}, who independently developed the theory
in the 1940s), with both the Chapman-Jouguet case and the overdriven
solutions present. Furthermore, stability analysis and unsteady simulations
of the model demonstrate the complexity seen in galloping detonations,
in particular their chaotic dynamics. These findings suggest that
a theory much simpler than the full reactive Euler equations may be
capable of describing the rich shock dynamics observed in detonation
waves. 

Simplified models have been used in the past to study detonations.
Both rational asymptotic theories and \emph{ad hoc} models have been
introduced previously to gain insight into the dynamics of detonation.
The reader can find extensive references in the recent review articles
and books \cite{zhang2012shock,Lee-2008,clavin2012analytical}. The
most relevant to our study is the theory of weakly nonlinear detonations
\cite{RosalesMajda:1983ly}, which is a model derived asymptotically
from the reactive Euler equations. Before \cite{RosalesMajda:1983ly},
Fickett \cite{Fickett:1979ys} and Majda \cite{Majda:1980zr} independently
introduced \emph{ad hoc} analog models, which were based on the idea
of extending Burgers' equation by an additional equation modeling
chemical reactions. The effect of chemical reactions in these analogs
appears as a modification of the flux function to include the chemical
energy term. The analog models received much attention in the past
\cite{Fickett:1979ys,fickett1985introduction,fickett1985stability,fickett1987decay,fickett1989approach,Radulescu:2011fk}
and continue to attract interest from a mathematical point of view
\cite{humpherys2013stability}. These simplified models possess a
theory analogous to steady ZND theory, with its Chapman-Jouguet (CJ),
strong, and weak detonation solutions. The weakly non-linear model
\cite{RosalesMajda:1983ly} is a result of an asymptotic reduction
of the reactive Euler equations. It applies in any number of spatial
dimensions, reducing in one dimension to equations very similar to
those of the analogs and therefore also containing the theory of steady
ZND waves. The analog models have been thought to perform poorly in
describing galloping one-dimensional instabilities and the transition
to chaos. However, the recent work of Radulescu and Tang \cite{Radulescu:2011fk}
demonstrates that a slightly modified version of Fickett's analog,
to include a two-stage chemical reaction with an inert induction zone
and a following reaction zone, reproduces much of the complexity of
detonations in reactive Euler equations. We suggest that even a much
simpler scalar equation can capture many of the known phenomena of
pulsating detonation waves.

The remainder of this paper is structured as follows. In Section 2,
we introduce the model and discuss its connection with the weakly
nonlinear model. Next, we develop a general theory for the proposed
equation and compute the possible steady ZND solutions. In Section
3, we derive a dispersion relation for the linear stability, and prove
certain important properties about the distribution of the eigenvalues.
Finally, in Section 4, we focus on a specific example, for which we
perform an extensive numerical study. With the example, we calculate
the linear stability spectrum, the onset of instabilities, and the
long-time nonlinear dynamics of solutions. Using tools from dynamical
system theory, we show that the solution goes through a sequence of
period doubling bifurcations to chaos, much like in the reactive Euler
equations.

\section{The Model}

Our model construction is based on two basic ideas: weakly nonlinear
approximation \cite{RosalesMajda:1983ly} and non-locality of the
chemical energy release rate \cite{fickett1985introduction}. The
precise nature of this non-locality is explained below. The weakly
nonlinear theory of detonation in one dimension, in the inviscid limit,
results in the following simplified system \cite{RosalesMajda:1983ly}:
\begin{alignat}{1}
 & u_{t}+\left(\frac{u^{2}}{2}+\frac{q}{2}\lambda\right)_{\eta}=0,\label{eq:mod_ros1}\\
 & \lambda_{\eta}=\omega\left(\lambda,u\right),\label{eq:mod_ros3}
\end{alignat}
where $t$ and $\eta$ are time and spatial variables, respectively;
$\lambda$ is the mass fraction of reaction products, going from $0$
ahead of the shock to $\text{1}$ in the fully burnt mixture; $u$
can be thought of as, for example, a temperature; $\omega(\lambda,u)$
is the reaction rate and $q$ is a constant representing the chemical
heat release. Note that (\ref{eq:mod_ros3}) propagates waves instantaneously
since the time derivative is missing in the equation. Nevertheless,
(\ref{eq:mod_ros1}--\ref{eq:mod_ros3}) constitutes a hyperbolic
system. 

Consider a shock moving into an unreacted ($\lambda=0$), unperturbed
($u=0$) region. At the shock, we apply the Rankine-Hugoniot conditions
to (\ref{eq:mod_ros1}) to obtain 

\begin{equation}
-D\left[u\right]+\frac{1}{2}\left[u^{2}\right]+\frac{q}{2}\left[\lambda\right]=0,\label{eq:R-H}
\end{equation}
where $D$ is the shock speed and the brackets denote the jump across
the shock in the enclosed variables. Using $\left[\lambda\right]=0$
and that $u=0$ ahead of the shock, it follows from (\ref{eq:R-H})
that $D=\dot{\eta}_{s}=u_{s}/2$, where $\eta_{s}(t)$ is the shock
position and $u_{s}=u\left(\eta_{s}^{-},t\right)$ denotes the post-shock
value of $u$. A change of variables to the shock-attached frame,
given by $x=\eta-\eta_{s}(t)$, yields 
\begin{alignat}{1}
 & u_{t}+\left(\frac{u^{2}}{2}+\frac{q}{2}\lambda-Du\right)_{x}=0,\label{eq:mod_ros1-1}\\
 & \lambda_{x}=\omega\left(\lambda,u\right),\label{eq:mod_ros3-1}
\end{alignat}
for $x\leq0$ and $u=0,\ \lambda=0$ for $x>0$. 

Now we make an important assumption that $\omega(\lambda,u)=\omega(\lambda,u_{s})$.
This simplifying assumption is the reason why we call the model nonlocal,
because the change of $\lambda$ at any given point $x$ at time $t$
is determined not by $u\left(x,t\right)$ at that point, but by $u$
at the shock, $x=0$. This means that any change of $u_{s}\left(t\right)$
propagates instantaneously over the whole domain, $x<0$. Note that
such assumption is sometimes used in modeling detonation in condensed
explosives. The idea behind it is that the energy release is primarily
controlled by how hard the explosive is hit by the shock \cite{wackerle1978shock,fickett1985introduction}. 

As a consequence of the assumed form of $\omega$, equation (\ref{eq:mod_ros3-1})
can now be integrated over $x$ to yield $\lambda=F(x,u_{s})$. Upon
differentiation of the latter with respect to $x$ and substitution
into (\ref{eq:mod_ros1-1}) (letting $qF_{x}/2=f$), we obtain one
non-local equation on the half-line, $x\leq0$, given by
\begin{equation}
u_{t}+\frac{1}{2}\left(u^{2}-uu_{s}\right)_{x}=f\left(x,u_{s}\right).\label{eq:KFR-equation}
\end{equation}
Conversely, it can be shown that for any positive function, $f$,
a function $\omega(\lambda,u_{s})$ can be found such that (\ref{eq:KFR-equation})
is equivalent to the system given by (\ref{eq:mod_ros1-1})-(\ref{eq:mod_ros3-1}). 

The shock, which is now located at $x=0$ at any $t$, must satisfy
the Lax conditions, that is, $c(0^{-},t)>0>c(0^{+},t)$, where $c=u-u_{s}/2$
denotes the characteristic speed in (\ref{eq:KFR-equation}). It follows
that $D\left(t\right)=u_{s}/2=c\left(0^{-},t\right)>0$.

Initial data for (\ref{eq:KFR-equation}) are given as $u\left(x,0\right)=g\left(x\right)$
for $x<0$, where $g\left(x\right)$ is a suitable function and $u\left(x,0\right)=0$
for $x>0$ is assumed implicitly. An important feature of (\ref{eq:KFR-equation})
is that the boundary value of the unknown, $u_{s}$, is contained
within the equation. This is one of the key reasons for the observed
complexity of the shock dynamics. While the boundary information from
the shock at $x=0$ is propagated instantaneously throughout the solution
domain at $x<0$, there is a finite-speed influence propagating from
the reaction zone toward the shock along the characteristics of (\ref{eq:KFR-equation}).

In characteristic form, (\ref{eq:KFR-equation}) can be written as
\begin{alignat}{1}
 & \frac{du}{d\tau}=f\left(x,u_{s}\right),\\
 & \frac{dx}{d\tau}=u-\frac{u_{s}}{2},
\end{alignat}
where the characteristic speed is $c=u-u_{s}/2$. Therefore, (\ref{eq:KFR-equation})
incorporates, within a single scalar equation, the nonlinear interaction
of two waves. One is the usual Burgers wave propagating toward the
shock at a finite speed, $c$. The other is of an unusual type, as
it represents an instantaneous effect by the state $u_{s}$ at the
shock, $x=0$, on the whole solution region $x<0$. Physically, this
second wave corresponds to the particle paths carrying the reaction
variable as explained in \cite{kasimov2013model}. In the weakly nonlinear
limit, these paths have, effectively, an infinite velocity.

\section{Steady solutions and their stability}

In this section, we explore some general properties of the proposed
model. Keeping in mind the connection with detonation theory, we restrict
our attention to $f(x,u_{s})$ such that $\int_{-\infty}^{0}f(x,u_{s})\ dx=q/2=const$.
This condition means that the amount of energy released by reactions
is finite and fixed. We consider only exothermic reactions; hence,
$f(x,u_{s})\geq0$. Although these assumptions facilitate some of
the computations, they are not required for most of the results presented
here, and more general forms of the forcing can be considered without
adding much more complexity to the analysis.

\subsection{\label{sub:Steady-state-solutions}Steady state solutions}

Let $u_{0}\left(x\right)$ denote a steady-state smooth solution of
(\ref{eq:KFR-equation}). It is a solution of 
\begin{equation}
\frac{d}{dx}\left(\frac{u_{0}^{2}}{2}-\frac{u_{0}u_{0s}}{2}\right)=f\left(x,u_{0s}\right),\label{eq:u_0_ode}
\end{equation}
or, equivalently, 
\[
\left(u_{0}-\frac{u_{0s}}{2}\right)u_{0}'=f\left(x,u_{0s}\right),
\]
where `` $'$ '' denotes the derivative with respect to $x$ and
$u_{0s}=u_{0}\left(0\right)$ is the steady-state value of $u$ at
$x=0$, which is to be found together with $u_{0}\left(x\right)$.
Integration of (\ref{eq:u_0_ode}) from $0$ to $x$ yields a quadratic
equation for $u_{0}$, 
\[
u_{0}^{2}-u_{0}u_{0s}=2\int_{0}^{x}f\left(y,u_{0s}\right)dy,
\]
where the integration constant vanishes in view of the boundary condition
at $x=0$. The solution profile is thus given by 
\begin{equation}
u_{0}\left(x\right)=\frac{u_{0s}}{2}+\sqrt{\frac{u_{0s}^{2}}{4}+2\int_{0}^{x}f\left(y,u_{0s}\right)dy}.\label{eq:u_0(x)}
\end{equation}
The plus sign is chosen here to satisfy the boundary condition at
$x=0$. We note that for $u_{0}\left(x\right)$ in (\ref{eq:u_0(x)})
to be real, $f$ must be constrained so that at any $x$, the expression
under the square root is non-negative. Effectively, this is the requirement
of overall exothermicity of the source term.

The choice of $u_{0s}$ depends on the behavior of the solution at
$x\to-\infty$. For the square root in equation (\ref{eq:u_0(x)})
to be real at $x=-\infty$, we require that 
\begin{equation}
u_{0s}=\zeta\left(2\sqrt{2\int_{-\infty}^{0}f\left(y,u_{0s}\right)dy}\right)\label{eq:u_0s_overdriven}
\end{equation}
with some $\zeta\geq1$. The effect of $\zeta$, which is the analog
of the overdrive factor in detonation theory, on the shape and the
stability of the traveling wave can be readily appreciated in the
non-dimensional formulation given in Section \ref{sec:An-example}.
The case with $\zeta=1$ whereby 
\begin{equation}
u_{0s}=2\sqrt{2\int_{-\infty}^{0}f\left(y,u_{0s}\right)dy},\label{eq:u_0s_cj}
\end{equation}
is an important special case commonly referred to as the Chapman-Jouguet
solution, because the characteristic speed at $x=-\infty$ is $c_{0}\left(-\infty\right)=u_{0}\left(-\infty\right)-u_{0s}/2=0$.
Therefore, the characteristics point toward the shock everywhere at
$x<0$ becoming vertical at $x=-\infty$. Cases where $\zeta>1$ are
related to piston-driven detonations wherein the state at $x=-\infty$
remains subsonic, i.e., $c>0$. In the context of the Euler detonations,
they are known to be more stable than Chapman-Jouguet waves \cite{LeeStewart90,short1998cellular}.

\subsection{\label{sub:Spectral-stability}Spectral stability of the steady-state
solution}

Consider the linear stability of the steady-state solution obtained
in the previous section. For simplicity, we limit the analysis to
the CJ case, but the overdriven solution can be similarly analyzed.
Let $u\left(x,t\right)=u_{0}\left(x\right)+\epsilon u_{1}\left(x,t\right)+O\left(\epsilon^{2}\right)$
with $\epsilon\to0$ and linearize (\ref{eq:KFR-equation}). We find
that 
\begin{equation}
u_{1t}+\left(u_{0}-\frac{u_{0s}}{2}\right)u_{1x}+u_{0}'u_{1}=\left(\frac{\partial f}{\partial u_{s}}\left(x,u_{0s}\right)+\frac{u_{0}'}{2}\right)u_{1}\left(0,t\right).
\end{equation}
The steady-state characteristic speed is 
\begin{equation}
c_{0}=u_{0}-\frac{u_{0s}}{2}=\sqrt{2\int_{-\infty}^{x}f\left(y,u_{0s}\right)dy},
\end{equation}
and the coefficient on the right-hand side of the linearized equation
above is 
\begin{equation}
b_{0}\equiv\frac{\partial f}{\partial u_{s}}\left(x,u_{0s}\right)+\frac{u_{0}'}{2}=\frac{\partial f}{\partial u_{s}}\left(x,u_{0s}\right)+\frac{f\left(x,u_{0s}\right)}{2c_{0}\left(x\right)}=\frac{\partial f}{\partial u_{s}}\left(x,u_{0s}\right)+\frac{1}{2}c_{0}(x)'.
\end{equation}
Both $c_{0}$ and $b_{0}$ are functions of $x$. 

Thus, the linear stability problem requires that the following linear
non-local PDE with variable coefficients, 
\begin{equation}
u_{1t}+c_{0}u_{1x}+c_{0}'u_{1}=b_{0}u_{1}\left(0,t\right),\label{eq:u_1_pde}
\end{equation}
be solved subject to appropriate initial data, $u_{1}\left(x,0\right)$.
If spatially bounded (in some norm, to be defined below) solutions
of (\ref{eq:u_1_pde}) grow in time, then instability is obtained.
At this point, we can proceed with either the Laplace transform in
time (as in \cite{Erpenbeck62}) or normal modes (as in \cite{LeeStewart90}).
We choose the latter and substitute the normal modes, 
\begin{equation}
u_{1}=\exp\left(\sigma t\right)v\left(x\right),\label{eq:normal-mode}
\end{equation}
into (\ref{eq:u_1_pde}), to obtain 
\[
c_{0}v'+c_{0}'v+\sigma v=b_{0}\left(x\right)v\left(0\right).
\]
This equation can be integrated directly to yield 
\[
\exp\left(\sigma\int_{0}^{x}\frac{dy}{c_{0}\left(y\right)}\right)c_{0}\left(x\right)v\left(x\right)-c_{0}\left(0\right)v\left(0\right)=v\left(0\right)\int_{0}^{x}b_{0}\left(\xi\right)\exp\left(\sigma\int_{0}^{\xi}\frac{dy}{c_{0}\left(y\right)}\right)d\xi.
\]
Denoting $p=\int_{x}^{0}dy/c_{0}\left(y\right)>0$, we obtain the
final solution for the amplitude of the normal mode as 
\begin{equation}
{\displaystyle v\left(x\right)=v\left(0\right)p'\left(x\right)e^{\sigma p\left(x\right)}\left[\int_{x}^{0}b_{0}\left(\xi\right)e^{-\sigma p\left(\xi\right)}d\xi-c_{0}\left(0\right)\right].}\label{eq:v(x)}
\end{equation}
The existence of an unstable eigenvalue with $\mathbb{\Re}(\sigma)>0$
and bounded $v(x)$ is equivalent to normal-mode instability. On physical
grounds, we require that $f$ be integrable in $x$ at any given $t$
(i.e., the $L^{1}$ norm of $f$ is bounded). This requirement follows
from the implicit assumption that $f$ is in fact the $x-$derivative
of some reaction progress variable, $\lambda$, varying between $0$
and $1$. We impose the same constraint on $u$, hence $v\in L^{1}\left(\mathbb{R}^{-}\right)$.

Note that $p\left(x\right)\to\infty$ as $x\to-\infty$, therefore,
the factor in front of the brackets in (\ref{eq:v(x)}) tends to infinity
as $x\to-\infty$. To prevent this super-exponential growth, the term
in the brackets must vanish as $x\to-\infty$. In fact, this condition
is also sufficient for instability.
\begin{thm}
\label{thm:disp-relation}Provided that $\|b_{0}(x)\|_{L^{1}}<\infty$,
the existence of a $\sigma$ with $\Re(\sigma)>0$ such that 
\begin{equation}
\int_{-\infty}^{0}b_{0}\left(\xi\right)e^{-\sigma p\left(\xi\right)}d\xi-c_{0}\left(0\right)=0,\label{eq:dispersion-relation}
\end{equation}
is both necessary and sufficient for the existence of unstable normal
modes, (\ref{eq:normal-mode}). \end{thm}
\begin{proof}
If condition (\ref{eq:dispersion-relation}) is not satisfied, then
$v(x)\rightarrow\infty$ as $x\rightarrow-\infty$. Now, suppose that
(\ref{eq:dispersion-relation}) is satisfied. Then, $v(x)$ takes
the form 
\begin{equation}
v\left(x\right)=\frac{v\left(0\right)}{c_{0}\left(x\right)}\int_{-\infty}^{x}b_{0}\left(\xi\right)e^{-\sigma\left(p\left(\xi\right)-p\left(x\right)\right)}d\xi.\label{eq:v(x)-1}
\end{equation}
We now show that $\Vert v\left(x\right)\Vert_{L^{1}}<\infty$. From
(\ref{eq:v(x)-1}), it follows that

\begin{alignat*}{1}
 & \Vert v\Vert_{L^{1}}=\left|v\left(0\right)\right|\int_{-\infty}^{0}dx\frac{1}{\left|c_{0}\left(x\right)\right|}\left|\int_{-\infty}^{x}b_{0}\left(\xi\right)e^{-\sigma\left(p\left(\xi\right)-p\left(x\right)\right)}d\xi\right|\le\\
 & \left|v\left(0\right)\right|\int_{-\infty}^{0}d\xi\int_{\xi}^{0}dx\frac{1}{\left|c_{0}\left(x\right)\right|}\left|b_{0}\left(\xi\right)\right|e^{-\Re\left(\sigma\right)\left(p\left(\xi\right)-p\left(x\right)\right)}.
\end{alignat*}
We change the integration variable in the inner integral from $x$
to $z=p\left(\xi\right)-p\left(x\right)$, so that $dx=-dz/p'\left(x\right)=c_{0}\left(x\right)dz$.
Then, 

\begin{alignat}{1}
 & \Vert v\Vert_{L^{1}}\le\left|v\left(0\right)\right|\int_{-\infty}^{0}d\xi\int_{0}^{p\left(\xi\right)-p\left(0\right)}dz\left|b_{0}\left(\xi\right)\right|e^{-\Re\left(\sigma\right)z}\le\frac{\left|v\left(0\right)\right|}{\Re\left(\sigma\right)}\Vert b_{0}\Vert_{L^{1}},
\end{alignat}
which proves that the unstable perturbations are bounded in the $L^{1}$
norm provided that $b_{0}\in L^{1}\left(\mathbb{R}^{-}\right)$. Thus,
(\ref{eq:dispersion-relation}) is both necessary and sufficient for
the existence of unstable normal modes. \end{proof}
\begin{rem*}
The dispersion relation (\ref{eq:dispersion-relation}) closely resembles
that of \cite{clavin1996acoustic,clavin1997multidimensional}, where
the detonation dynamics is analyzed in the asymptotic limit of strong
overdrive. In this limit, the entire flow downstream of the lead shock
has a small Mach number relative to the shock, hence the post-shock
pressure remains nearly constant. For this reason, such approximation
is called quasi-isobaric. However, the underlying assumptions in the
present model and those in the quasi-isobaric theory are quite different.
\end{rem*}
Another important result is that, under appropriate assumptions on
$f$, the unstable modes have a bounded growth rate. This result shows
that the so-called ``pathological'' instability, inherent to square-wave
models of detonation in the Euler equations \cite{zaidel1961stability,erpenbeck1963squarewave,fickett1985stability,he1999two},
does not occur in our model for smooth steady-state solutions. However,
in section \ref{sub:The-square-wave-limit} we show that this pathological
instability occurs in the square-wave limit of our model, when $f$
is replaced by a delta function. 
\begin{thm}
\label{thm:no-large-realpart-eig}Provided that $\|b_{0}c_{0}\|_{L^{\infty}}=M<\infty$,
there exist no eigenvalues with $\sigma_{r}>M/c_{0}\left(0\right)$. \end{thm}
\begin{proof}
Notice that 
\begin{eqnarray*}
\left|\int_{-\infty}^{0}b_{0}(x)e^{-\sigma p(x)}dx\right| & \leq & \int_{-\infty}^{0}\left|b_{0}(x)e^{-\sigma p(x)}\right|dx\\
 & = & \int_{-\infty}^{0}\left|b_{0}(x)e^{-\sigma_{r}p(x)}\right|dx.
\end{eqnarray*}
Let $z=p(x)$ and note that this function is invertible since $p$
is monotonic. Substitution into the previous integral yields
\begin{eqnarray*}
\int_{-\infty}^{0}|b_{0}(x)e^{-\sigma_{r}p(x)}|dx & = & \int_{0}^{\infty}|b_{0}(p^{-1}(z))c_{0}(p^{-1}(z))|e^{-\sigma_{r}z}dx\\
 & \leq & \max_{-\infty\leq x\leq0}|b_{0}c_{0}|\int_{0}^{\infty}e^{-\sigma_{r}z}dx\\
 & = & \frac{1}{\sigma_{r}}\max_{-\infty\leq x\leq0}|b_{0}c_{0}|,
\end{eqnarray*}
and thus for 
\[
\sigma_{r}>\frac{\max_{\infty\leq x\leq0}|b_{0}c_{0}|}{c_{0}(0)},
\]
we obtain 
\[
\left|\int_{-\infty}^{0}b_{0}(x)e^{-\sigma p(x)}dx\right|\leq\frac{1}{\sigma_{r}}\max_{\infty\leq x\leq0}|b_{0}c_{0}|<c_{0}(0).
\]
 This contradicts the dispersion relation stated in Theorem 1. \end{proof}
\begin{rem*}
If $f(x,u_{0s})$ is integrable and bounded and $\frac{\partial f}{\partial u_{s}}\left(x,u_{0s}\right)$
is bounded, then it can be shown that $b_{0}c_{0}\in L^{\infty}$.
These constraints are sufficient to eliminate the pathological instabilities
in which arbitrarily large growth rates are present. \end{rem*}
\begin{thm}
If $\|b_{0}c_{0}\|_{L^{\infty}}=M<\infty$, there exists a bounded
interval $I$ large enough that all eigenvalues with $\sigma_{r}>0$
have imaginary part $|\sigma_{i}|<I$.\end{thm}
\begin{proof}
By application of the Riemann-Lebesgue lemma, we find that 
\begin{eqnarray*}
\int_{-\infty}^{0}b_{0}(x)e^{-\sigma p(x)}dx & = & \int_{0}^{\infty}b_{0}(x)c_{0}(x)e^{-\sigma z}dz\\
 & = & \int_{0}^{\infty}\left(b_{0}(x)c_{0}(x)e^{-\sigma_{r}z}\right)e^{i\sigma_{i}z}dx\to0\mbox{ as }\sigma_{i}\rightarrow\infty
\end{eqnarray*}
provided that $b_{0}(p^{-1}(z))c_{0}(p^{-1}(z))e^{-\sigma_{r}z}\in L^{1}$.
If $\sigma_{r}>0$ and $b_{0}c_{0}$ is bounded, then it follows that
indeed $b_{0}(p^{-1}(z))c_{0}(p^{-1}(z))e^{-\sigma_{r}z}\in L^{1}$.
Therefore, the integral above vanishes as $\sigma_{i}\to\infty$,
which cannot happen because the integral should equal to $c_{0}\left(0\right)=u_{0s}/2>0$. \end{proof}
\begin{thm}
\label{thm:no-zero-eigenvalues}$\sigma=0$ is never an eigenvalue. \end{thm}
\begin{proof}
The condition $\int_{-\infty}^{0}b_{0}\left(\xi\right)e^{-\sigma p\left(\xi\right)}d\xi-c_{0}\left(0\right)=0$
is still necessary for the eigenfunctions to remain bounded, even
when $\sigma=0$. Therefore, $\int_{-\infty}^{0}b_{0}\left(\xi\right)d\xi-c_{0}\left(0\right)=0$,
or equivalently 
\begin{eqnarray*}
\int_{-\infty}^{0}\left[\frac{\partial f}{\partial u_{s}}\left(\xi,u_{0s}\right)+\frac{1}{2}c_{0}(\xi)'\right]d\xi-c_{0}\left(0\right) & = & 0,\\
\int_{-\infty}^{0}\frac{\partial f}{\partial u_{s}}\left(\xi,u_{0s}\right)d\xi & = & c_{0}(0)/2.
\end{eqnarray*}
Since we assume that $f$ integrates to a constant, then 
\begin{eqnarray*}
\int_{-\infty}^{0}\frac{\partial f}{\partial u_{s}}\left(\xi,u_{0s}\right)d\xi & = & \frac{d}{du_{s}}\int_{-\infty}^{0}f\left(\xi,u_{0s}\right)d\xi=0.
\end{eqnarray*}
But $c_{0}(0)=u_{0s}/2>0$, and therefore no such eigenvalue can exist.
Thus, at the onset of instability, the eigenvalues must have non-zero
frequency. 
\end{proof}
Because $\sigma=0$ is never an eigenvalue, when the behavior of the
system as a function of parameters is explored, the transition from
a stable steady state to instability usually involves a Hopf bifurcation.
In our numerical calculations we find that this bifurcation is a supercritical
Hopf bifurcation, so that a stable time periodic solution takes over
from the steady state.

\section{\label{sec:An-example}An example}

In the previous section, we presented necessary and sufficient conditions
for the normal-mode instability of a traveling wave profile. We now
focus on a specific choice of $f(x,u_{s})$ and illustrate with it
the general results on the linear instability. We also examine, by
means of direct numerical simulations, what happens once the traveling-wave
solution becomes unstable as a bifurcation parameter is varied. The
example mimics, on a qualitative level, a situation wherein the chemical
reaction has an induction zone that delays the beginning of an energetic
exothermic reaction. The idea is to have a function that peaks at
some distance away from the shock, with this distance depending on
the shock strength. A simple choice for such a function is 
\[
f=\frac{q}{2}\frac{1}{\sqrt{4\pi\beta}}\exp\left[-\frac{\left(x-x_{i}\left(u_{s}\right)\right)^{2}}{4\beta}\right].
\]
Here, $x_{i}$ is the point where $f$ peaks and that point depends
on the current state at the shock, $u_{s}=u\left(0,t\right)$. The
parameter $\beta$ determines the width of the reaction zone. As $\beta\to0$,
$f$ tends to $\frac{q}{2}\delta\left(x-x_{i}\right)$; this limit
yields what is called a square-wave profile, wherein $f$ kicks in
only at $x=x_{i}$. We choose $x_{i}$ as 
\[
x_{i}\left(t\right)=-k\left(\frac{u_{0s}}{u_{s}\left(t\right)}\right)^{\alpha},
\]
which depends on the shock strength, $u_{s}$, the steady-state shock
strength, $u_{0s}$, and the parameters $k>0$ and $\alpha\ge0$.
Remembering the connection with the weakly nonlinear model, where
$f=q\lambda_{x}/2$, we require that 
\begin{equation}
\int_{-\infty}^{0}f(x,u_{s})\ dx=\frac{q}{2},\label{eq:f_integral}
\end{equation}
and thus renormalize $f$ as follows%
\footnote{Note that in \cite{kasimov2012model,kasimov2013model} $f$ was not
renormalized. %
}: 
\[
f\to\frac{q}{2}\frac{f}{\int_{-\infty}^{0}fdx}=\frac{q}{\left(1+\text{Erf}\left[\frac{k\left(\frac{u_{s}}{u_{0s}}\right)^{-\alpha}}{2\sqrt{\beta}}\right]\right)\sqrt{4\pi\beta}}\exp\left[-\frac{\left(x+k\left(\frac{u_{0s}}{u_{s}}\right)^{\alpha}\right)^{2}}{4\beta}\right].
\]
Next, the variables are rescaled as follows: $u=u_{0s}\tilde{u},\quad x=k\tilde{x},\quad t=k\tilde{t}/u_{0s},$
and $\beta=k^{2}\tilde{\beta}$, where the variables with the tildes
are now dimensionless. Using $u_{0s}=2\zeta\sqrt{q}$, that follows
from (\ref{eq:u_0s_overdriven}) and (\ref{eq:f_integral}), equation
(\ref{eq:KFR-equation}) takes the following dimensionless form 
\begin{equation}
\tilde{u}_{\tilde{t}}+\left(\frac{\tilde{u}^{2}}{2}-\frac{\tilde{u}\tilde{u}\left(0,\tilde{t}\right)}{2}\right)_{\tilde{x}}=\tilde{f}(\tilde{x},\tilde{u}_{s}),\label{eq:KFRexample_dimensionless}
\end{equation}
where
\begin{equation}
\tilde{f}(\tilde{x},\tilde{u_{s}})=\frac{1}{4\zeta^{2}\left(1+\text{Erf}\left[\frac{\tilde{u}(0,\tilde{t})^{-\alpha}}{2\sqrt{\tilde{\beta}}}\right]\right)}\frac{1}{\sqrt{4\pi\tilde{\beta}}}\exp\left[-\frac{\left(\tilde{x}+\left(\tilde{u}\left(0,\tilde{t}\right)\right)^{-\alpha}\right)^{2}}{4\tilde{\beta}}\right].\label{eq:forcing_dimensionless}
\end{equation}

This equation contains only three parameters, $\alpha$, which is
a measure of the shock-state sensitivity of the source function (analogous
to the activation energy in Euler detonations), $\tilde{\beta}=\beta/k^{2}$,
which is the width of $\tilde{f}$ (analogous to the ratio of the
reaction-zone length, $\sqrt{\beta},$ and the induction-zone length,
$k$), and $\zeta$, which is the overdrive factor. The role of the
latter is now easily appreciated: it scales the forcing term by $\text{\ensuremath{\zeta}}^{-2}$
such that the overdrive reduces the magnitude of the forcing and hence
has a stabilizing effect.

Our focus below is on the Chapman-Jouguet case, $\text{\ensuremath{\zeta}}=1$,
which leaves only $\alpha$ and $\beta$ as the parameters of the
model. Although the expression for the forcing is a little bit cumbersome,
its shape is simply that of a Gaussian shifted to the left of $x=0$
by $\tilde{u}(0,\tilde{t})^{-\alpha}$ and renormalized to integrate
to a constant on $(-\infty,0)$. A few examples of $\tilde{f}$ are
shown in Fig. \ref{fig:Plot-of-reaction-rate-and-steady-state}(a)
for different values of $u_{s}$ and fixed $\alpha$, $\beta$. The
main qualitative feature of $\tilde{f}$ is that it has a maximum
at some distance from $x=0$ and that the maximum is close to the
shock when $u_{s}$ is large and far from the shock when $u_{s}$
is small. These features mimic the behavior of the reaction rate in
Euler equations as a function of the lead-shock speed. 

\begin{figure}[h]
\centering{}\includegraphics[height=2.5in]{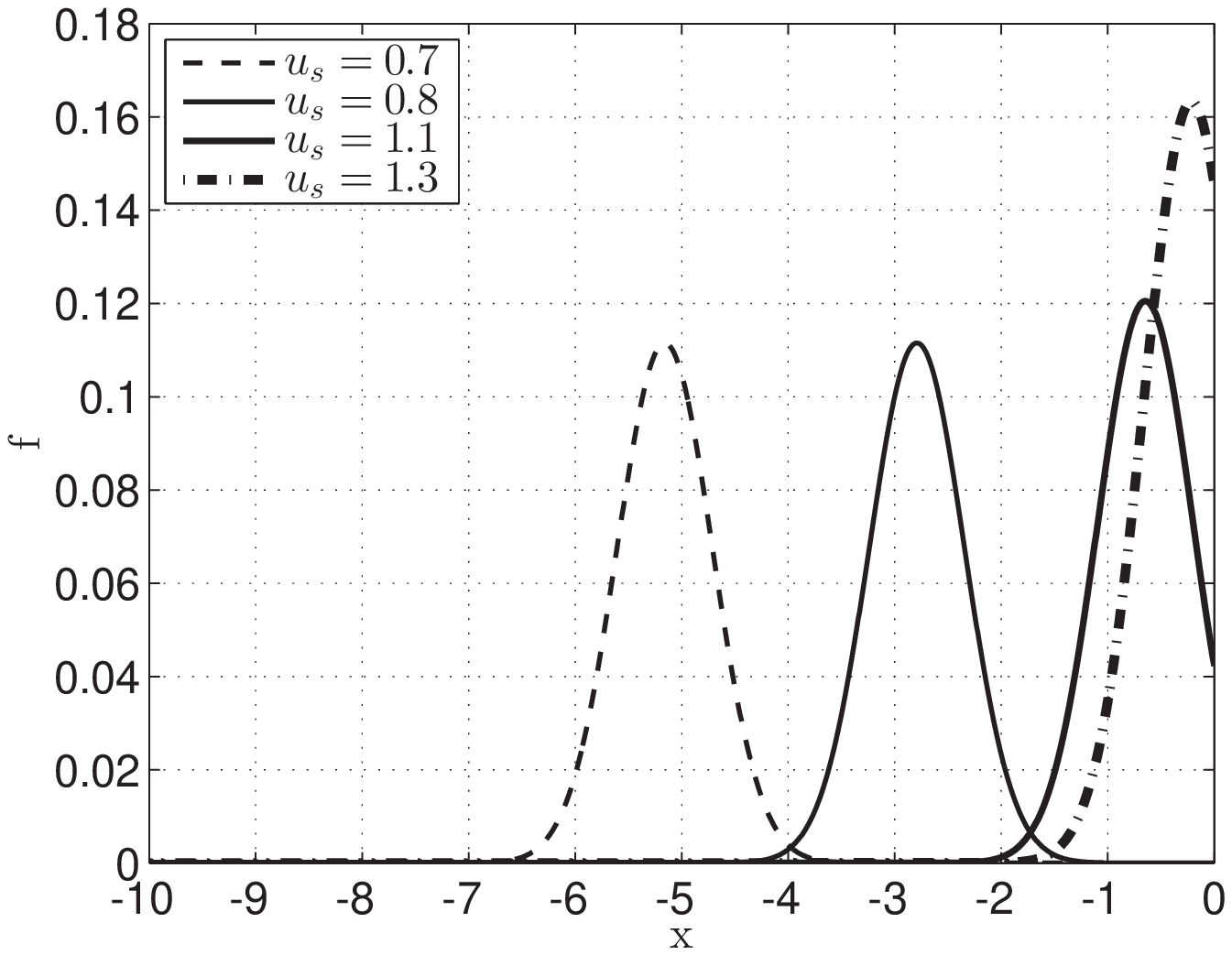}
\includegraphics[height=2.5in]{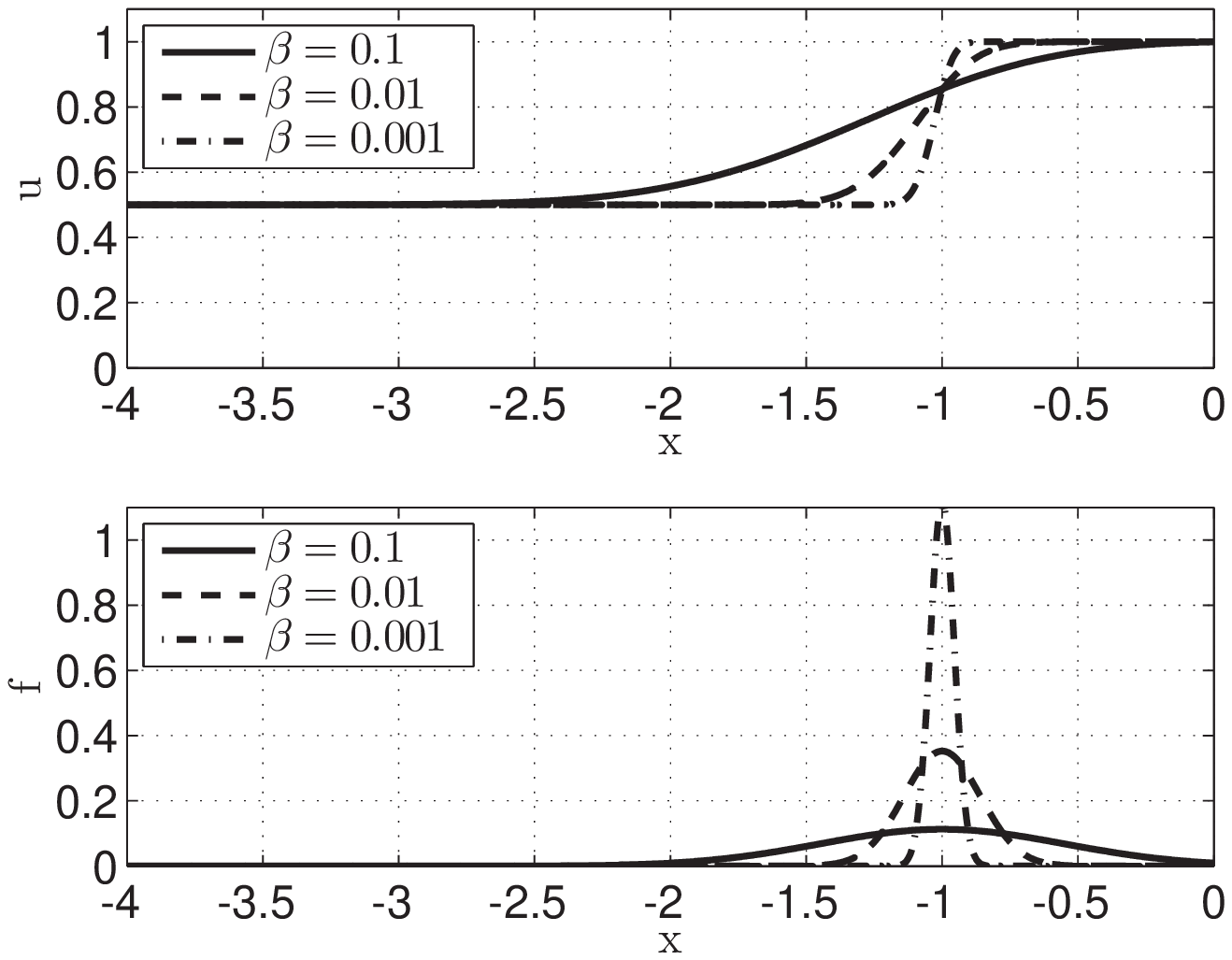}
\caption{\label{fig:Plot-of-reaction-rate-and-steady-state}(a) The forcing
term at various $u_{s}$. (b) Steady-state profiles and the forcing
function as $\beta$ is varied.}
\end{figure}

From now on, we drop the tilde notation, but it should be understood
that all the variables below are dimensionless.

\subsection{Steady-state solutions}

Steady-state Chapman-Jouguet solutions can be computed as shown in
Section \ref{sub:Steady-state-solutions}. Figure \ref{fig:Plot-of-reaction-rate-and-steady-state}(b)
shows how $\beta$ affects the traveling wave profile. The picture
suggests a square-wave solution in the limit $\beta\rightarrow0$. 

It is important to remember that $\alpha$ plays no role in the steady-state
profiles because $u_{0s}=1$ in dimensionless form. In some sense,
$\alpha$ represents the sensitivity to changes in the steady-state
profile. Next, we study the linear stability of these traveling wave
profiles in the $\alpha-\beta$ parameter space.

\subsection{Linear stability analysis}

\subsubsection{The dispersion relation}

By Theorem \ref{thm:disp-relation}, spectral instability is equivalent
to (\ref{eq:dispersion-relation}) provided $\Vert b_{0}\Vert_{L^{1}}<\infty$.
A straightforward computation shows that 
\[
\|b_{0}\|_{L^{1}}=\int_{-\infty}^{0}\left|\frac{\partial f}{\partial u_{s}}(x,u_{0s})+\frac{f(x,u_{0s})}{2c_{0}(x)}\right|dx<\infty,
\]
and therefore spectral stability of (\ref{eq:KFRexample_dimensionless})
is equivalent to 
\[
\int_{-\infty}^{0}b_{0}\left(\xi\right)e^{-\sigma p\left(\xi\right)}d\xi=c_{0}(0),
\]
where $b_{0},c_{0}$, and $p$ are defined as in Section \ref{sub:Spectral-stability}.
Although we have reduced the spectral stability of our problem to
finding complex roots of a single equation, the equation is (although
analytic in $\sigma$) numerically difficult. For a given $\alpha$
and $\beta$, an equation with three levels of nested integration
must be solved,

\begin{equation}
\int_{-\infty}^{0}\left\{ \left(\frac{\partial f\left(\xi,u_{0s}\right)}{\partial u_{s}}+\frac{\partial}{\partial\xi}\sqrt{\frac{1}{2}\int_{-\infty}^{\xi}f\left(y,u_{0s}\right)dy}\right)\exp\left[-\sigma\int_{\xi}^{0}\frac{dx}{\sqrt{2\int_{-\infty}^{x}f\left(y,u_{0s}\right)dy}}\right]\right\} d\xi-\frac{u_{0s}}{2}=0,\label{eq:eig-eq-expanded}
\end{equation}
where $\sigma=\sigma_{r}+i\sigma_{i}$ and $f$ is given by (\ref{eq:forcing_dimensionless}).
Interestingly, the original formulation of the linear stability problem
by Erpenbeck \cite{Erpenbeck62} requires the same three levels of
numerical integration (the steady-state solution, then the solution
of the adjoint homogeneous problem, and then the evaluation of the
dispersion relation). In general, these integrals require nearly machine-precision
evaluation of the functions in the integrands in order to obtain the
eigenvalues with only a few significant digits of accuracy. Except
for the limiting case of $\beta=0$, we find the roots numerically
using Matlab's \emph{fsolve} function, that uses a version of Newton's
method, and then we use Cauchy's argument principle to verify that
we have found all the \textcolor{black}{roots in a given region of
the complex plane. Here, Theorem 2 plays a fundamental role, since
it tells us that all eigenvalues must be within a finite region. When
$\beta=0$, we compute the roots analytically, and they serve as initial
guesses in the numerical continuation root-finding procedure when
$\beta$ is small.}

\subsubsection{\label{sub:The-square-wave-limit}The square-wave limit}

When $\beta\rightarrow0$, we obtain the square-wave solution. In
this limit, it can be shown that 
\[
\frac{\partial}{\partial u_{s}}f(x,u_{0s})=-\alpha\frac{\partial}{\partial x}f(x,u_{0s})+O\left(\frac{1}{\sqrt{\beta}}e^{-\frac{1}{4\beta}}\right)f(x,u_{0s}).
\]
Even though $f\left(x,u_{0s}\right)$ tends to a delta function when
$\beta\to0$, this function is integrated in the dispersion relation
and, therefore, the contribution of the second term above to the dispersion
relation is exponentially small in the limit due to the $O\left(\frac{1}{\sqrt{\beta}}e^{-\frac{1}{4\beta}}\right)$
factor. In the limit, the dispersion relation (\ref{eq:dispersion-relation})
becomes 
\begin{eqnarray*}
\int_{-\infty}^{0}b_{0}(x)\ e^{-\sigma p\left(x\right)}dx & = & \int_{-\infty}^{0}\left(\frac{\partial f}{\partial u_{s}}(x,u_{0s})+\frac{1}{2}\frac{\partial}{\partial x}(c_{0}(x))\right)e^{-\sigma p(x)}\ dx\\
 & = & \int_{-\infty}^{0}\left(-\alpha\frac{\partial f}{\partial x}(x,u_{0s})\right)e^{-\sigma p(x)}dx+\int_{-\infty}^{0}\left(\frac{1}{2}\frac{\partial}{\partial x}(c_{0}(x))\right)e^{-\sigma p(x)}\ dx.
\end{eqnarray*}
Integrating by parts, we find that 
\begin{align*}
-\alpha\int_{-\infty}^{0}\frac{\partial f}{\partial x}\left(x,u_{0s}\right)e^{-\sigma p(x)}dx+\frac{1}{2}\int_{-\infty}^{0}\frac{\partial}{\partial x}\left(c_{0}(x)\right)e^{-\sigma p(x)}dx & =c_{0}(0),\\
-\alpha\left[f(0,u_{0s})-\sigma\int_{-\infty}^{0}\frac{f\left(x,u_{0s}\right)}{c_{0}(x)}e^{-\sigma p(x)}dx\right]+\frac{1}{2}\int_{-\infty}^{0}\frac{\partial}{\partial x}\left(c_{0}(x)\right)e^{-\sigma p(x)}dx & =c_{0}(0),\\
-\alpha f(0,u_{0s})+\left(\alpha\sigma+\frac{1}{2}\right)\int_{-\infty}^{0}\frac{\partial}{\partial x}\left(c_{0}(x)\right)e^{-\sigma p(x)}dx & =c_{0}(0),\\
-\alpha f(0,u_{0s})+\left(\alpha\sigma+\frac{1}{2}\right)\left[c_{0}(0)-\sigma\int_{-\infty}^{0}e^{-\sigma p(x)}dx\right] & =c_{0}(0),\\
-\alpha f(0,u_{0s})+\alpha\sigma c_{0}(0)-\left(\alpha\sigma^{2}+\frac{\sigma}{2}\right)\int_{-\infty}^{0}e^{-\sigma p(x)}dx & =\frac{c_{0}(0)}{2},\\
-\alpha f(0,u_{0s})+\alpha\sigma c_{0}(0)-\left(\alpha\sigma^{2}+\frac{\sigma}{2}\right)\int_{0}^{\infty}c_{0}(x)e^{-\sigma z}dz & =\frac{c_{0}(0)}{2}.
\end{align*}
Noticing that
\[
c_{0}(x)\rightarrow\begin{cases}
\frac{1}{2} & x\geq-1\\
0 & x<-1
\end{cases}
\]
and 
\[
p(x)=\int_{x}^{0}\frac{1}{c_{0}(y)}dy\rightarrow\begin{cases}
\infty & x<-1\\
-2x & x\geq-1,
\end{cases}
\]
we obtain
\begin{align*}
\lim_{\beta\rightarrow0}\left[-\alpha f(0,u_{0s})+\alpha\sigma c_{0}(0)-\left(\alpha\sigma^{2}+\frac{\sigma}{2}\right)\int_{0}^{\infty}c_{0}\left(p^{-1}(z)\right)e^{-\sigma z}dz-\frac{c_{0}(0)}{2}+o(1)\right] & =\\
\frac{\alpha\sigma}{2}-\left(\frac{\alpha\sigma^{2}}{2}+\frac{\sigma}{4}\right)\int_{0}^{2}e^{-\sigma x}dx-\frac{1}{4} & =\\
\left(\frac{\alpha\sigma}{2}+\frac{1}{4}\right)e^{-2\sigma}-\frac{1}{2} & =0.
\end{align*}
Therefore, the dispersion relation in the square-wave limit takes
a very simple form of a transcendental equation 
\begin{equation}
e^{2\sigma}=\alpha\sigma+\frac{1}{2}.\label{eq:square-wave-dispersion}
\end{equation}
This dispersion relation has exactly the same form as that of Fickett's
analog \cite{fickett1985stability}, which in his case, arose from
his differential-difference equation for shock perturbation. Therefore,
it predicts the same pathological instability as in the classical
square-wave detonations. Pathological instability implies that the
linear stability problem for the square wave is ill-posed in the sense
of Hadamard. For completeness, we exhibit below the solutions to this
equation, since they are used as initial guesses in our algorithm
to compute the solutions when $\beta$ is small, but not zero. Let
$\sigma=\sigma_{r}+i\sigma_{i}$ and separate the real and imaginary
parts of (\ref{eq:square-wave-dispersion}), 
\begin{alignat*}{1}
 & e^{2\sigma_{r}}\cos\left(2\sigma_{i}\right)=\alpha\sigma_{r}+\frac{1}{2},\\
 & e^{2\sigma_{r}}\sin\left(2\sigma_{i}\right)=\alpha\sigma_{i}.
\end{alignat*}
If $\sigma_{r}$ is to be large, the first equation requires $\cos\left(2\sigma_{i}\right)$
to be small, i.e., $\sigma_{i}$ should be close to $\pi/4+n\pi/2,$
$n=0,1,2,...$. We let 
\[
\sigma_{i}=\frac{\pi}{4}+\frac{n\pi}{2}+\varepsilon,
\]
where $\varepsilon$ is a small correction. Then, from the second
equation, we find $\sin\left(2\sigma_{i}\right)\approx1$ and therefore
$\sigma_{r}\approx\frac{1}{2}\ln\left(\alpha\sigma_{i}\right)$. For
this $\sigma_{r}$ to be large, we need $n$ to be large, in which
case 
\[
\sigma_{r}\approx\frac{1}{2}\ln\left(n\right).
\]
Thus, the square-wave dispersion relation admits arbitrarily large
growth rates that occur at simultaneously large frequencies. It is
interesting that the growth rate increases with frequency logarithmically.
\textcolor{black}{Similar growth happens in the square-wave model
of detonations in the reactive Euler equations (see, e.g., \cite{zaidel1961stability,erpenbeck1963squarewave,buckmaster1988one,short1996asymptotic,short1997multidimensional,he1999two}).
However, in the latter, the dispersion relation involves several exponential
functions due to the presence of multiple time scales associated with
different families of waves propagating from the shock into the reaction
zone. Waves of different families of characteristics propagate at
different speeds resulting in several different time intervals for
the signals to propagate from the shock to the ``fire'' and back.
Since in the limit of large frequencies one of the exponentials dominates,
the dispersion relation becomes essentially the same as in our model.
In the numerical calculations of detonation instability in the Euler
equations with finite-rate chemistry, but high activation energies
\cite{short1997multidimensional}, a similarly slow growth can be
seen. However, we do not know if the growth is logarithmic in frequency. }
\begin{rem*}
Theorem \ref{thm:no-large-realpart-eig} is not contradicted here
since $\|b_{0}c_{0}\|\notin L^{\infty}$ in the limit, because now
$f\notin L^{\infty}$.
\end{rem*}

\subsubsection{The unstable spectrum for $\beta>0$}

The pathological instability of the model as $\beta\rightarrow0$
was shown to be caused by an infinite number of unstable eigenvalues,
with the real part arbitrarily large. From theorem \ref{thm:no-large-realpart-eig},
we know that if $\|b_{0}c_{0}\|_{L^{\infty}}=M<\infty$, then there
can be no unstable eigenvalues with $\sigma_{r}>M/c_{0}\left(0\right)$.
A quick computation shows that if $\alpha<\infty$ and $\beta>0$,
then the real part of the unstable spectrum of (\ref{eq:KFRexample_dimensionless})
is bounded from above.

Next, we fix $\alpha=4.05$ and numerically investigate the effect
of $\beta$ on the eigenvalues. Using as initial guess the eigenvalues
found from the square-wave dispersion relation, (\ref{eq:square-wave-dispersion}),
we use Matlab's numerical root finder, \emph{fsolve,} to locate the
eigenvalues for successively larger values of $\beta$. Figure \ref{fig:Spectrum-for-beta-varied}
shows the results, reaffirming that for any value of $\beta>0$, there
is only a finite number of unstable eigenvalues. Furthermore, it suggests
that the magnitude of $\beta$ is closely related to the frequencies
of the unstable eigenvalues. This can be understood as follows: as
the shock is perturbed, it creates waves that propagate into the reaction
zone. If $\beta$ is large enough, the reaction zone is smooth and
there is little resonance between the shock and the peak of the reaction
in the reaction zone. However, as $\beta$ is decreased, the sharp
peak in the reaction zone reflects waves back to the shock and this
resonance causes the instability. If $\beta$ is small but positive,
then high enough frequencies do not ``see'' the sharp peak in the
reaction rate and are not reflected back to the shock. 

\begin{figure}[h]
\begin{centering}
\subfloat[\label{fig:Spectrum-for-beta-varied}The spectrum for $\alpha=4.05$
with $\beta$ varied.]{\begin{centering}
\includegraphics[height=2.5in]{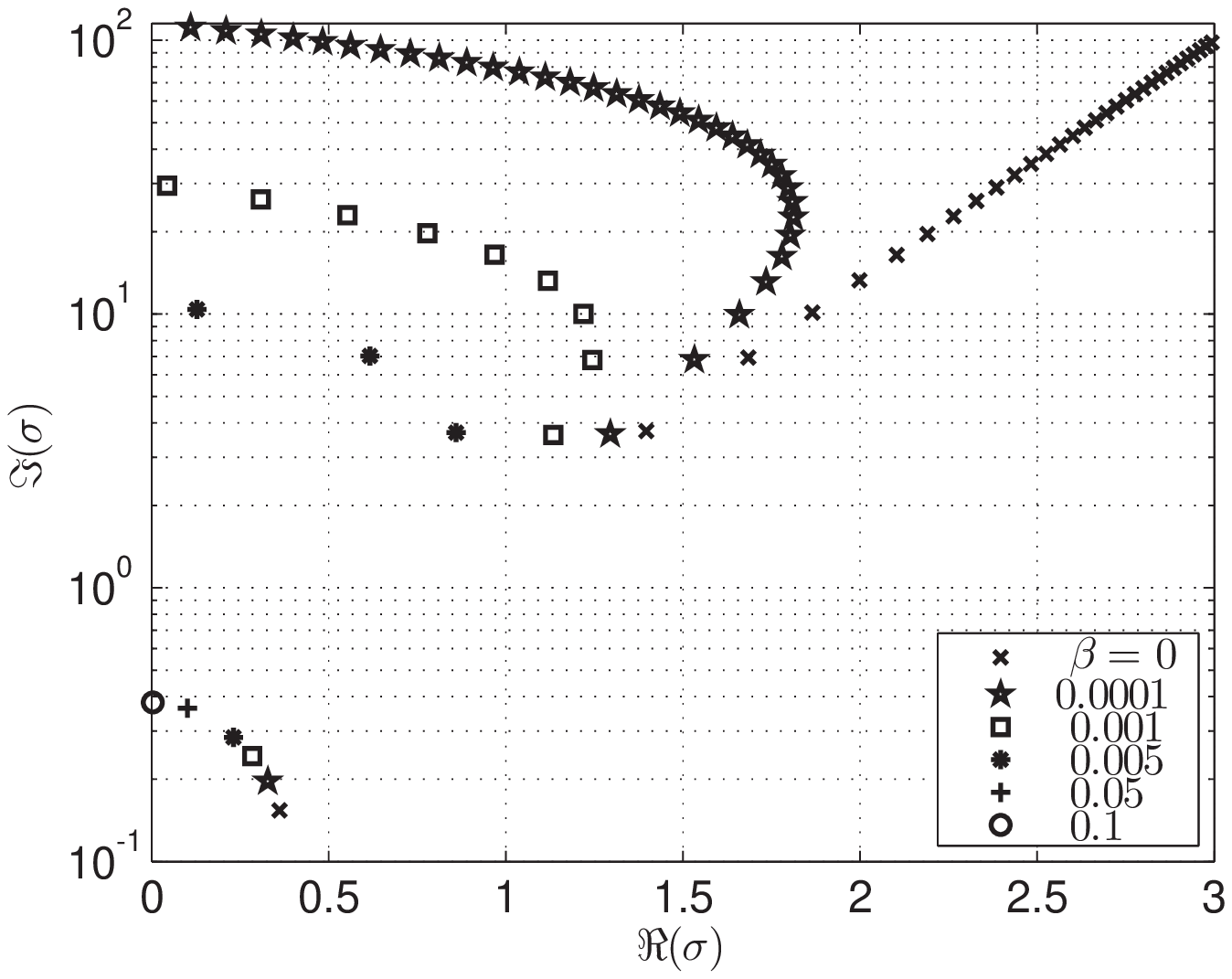}
\par\end{centering}

}\subfloat[\label{fig:Spectrum-for-alpha-varied}The spectrum for $\beta=0.001$
with $\alpha$ varied.]{

\includegraphics[height=2.5in]{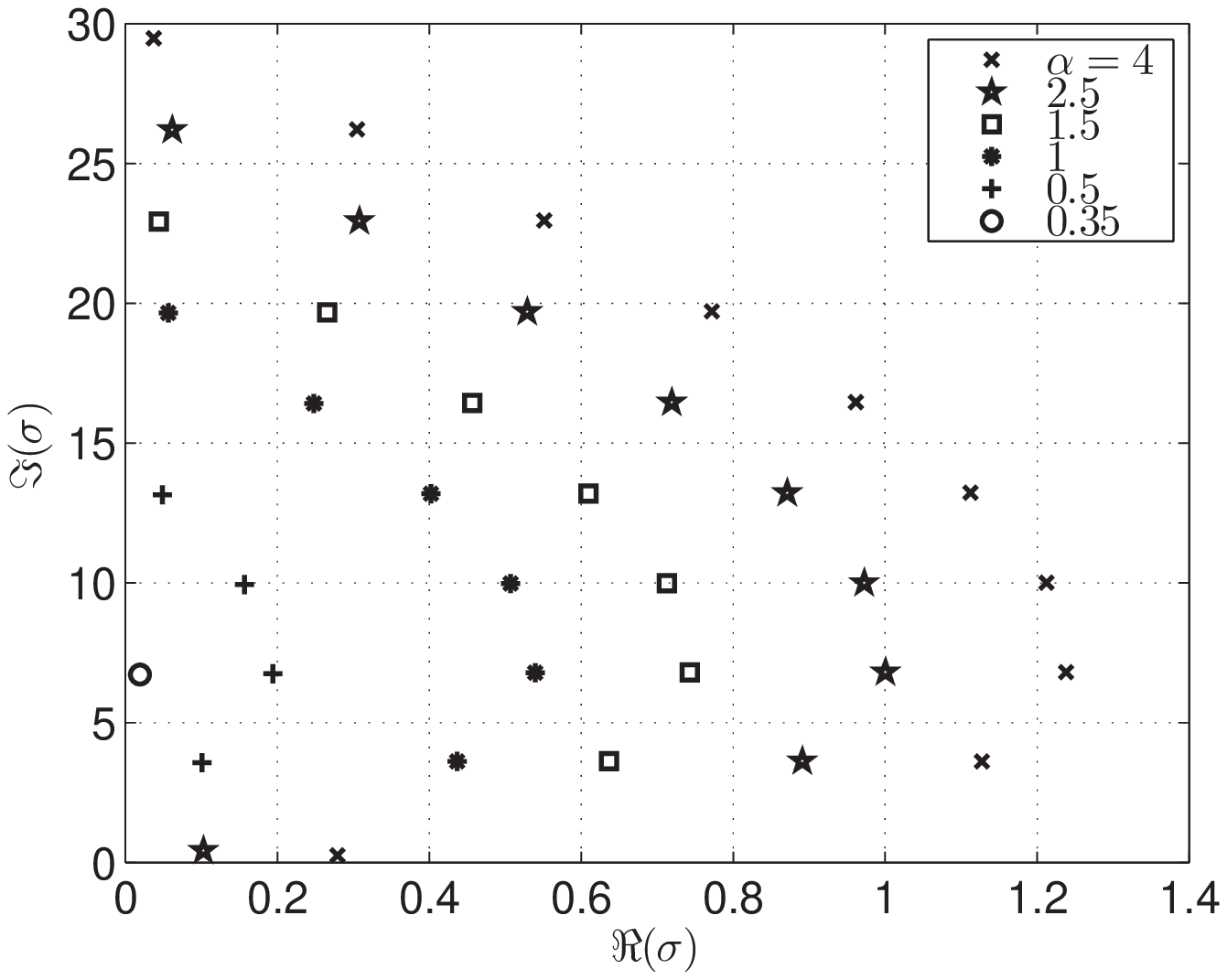}

}
\par\end{centering}

\caption{\label{fig:Spectrum}The linear spectrum.}
\end{figure}

We also look at the effect of $\alpha$ on the distribution of the
eigenvalues. In Fig. \ref{fig:Spectrum-for-alpha-varied}, we show
the spectrum for fixed $\beta=0.001$ and varying $\alpha$. This
figure suggests that the eigenvalues are merely shifted when $\alpha$
is decreased. Interestingly, the dominant eigenvalue, i.e. the one
with the largest real part, is always the same as we change $\alpha$
and keep $\beta$ fixed. This observation was tested for different
values of $\beta$. As $\beta$ decreases, the frequency of the most
unstable mode is seen to increase.

To ensure that no roots of the dispersion relation have been lost
in the numerical computations, we apply the argument principle to
(\ref{eq:dispersion-relation}). Since 
\[
F(\sigma)=\int_{-\infty}^{0}b_{0}\left(\xi\right)e^{-\sigma p\left(\xi\right)}d\xi-c_{0}(0)
\]
has no poles in the region $\sigma_{r}\geq0$ (which follows from
$\|b_{0}\|_{L_{1}}<\infty$), the argument principle guarantees that
the number of zeroes, $N$, of $F(\sigma)$ in a closed contour $C$
(counting multiplicity) is given by 
\[
N=\frac{1}{2\pi i}\int_{C}\frac{F'(z)}{F(z)}dz.
\]
This can be related to the winding number of a curve by the substitution
$w=F(z)$, which yields 
\[
N=\frac{1}{2\pi i}\int_{F(C)}\frac{1}{w}dw.
\]
We show in Fig. \ref{fig:argument-principle} two Nyquist plots of
the dispersion relation, corresponding to parameters with $2$ and
$20$ unstable eigenvalues. The predictions agree with the number
of roots found using the root solver. 

\begin{figure}[h]
\begin{centering}
\subfloat[\label{fig:Nyquist-alpha-4.05-beta-0.05}$\alpha=4.05$, $\beta=0.05$.
Weakly unstable case with two eigenvalues, one shown in Fig. \ref{fig:Spectrum}
and its complex conjugate.]{\includegraphics[height=2.5in]{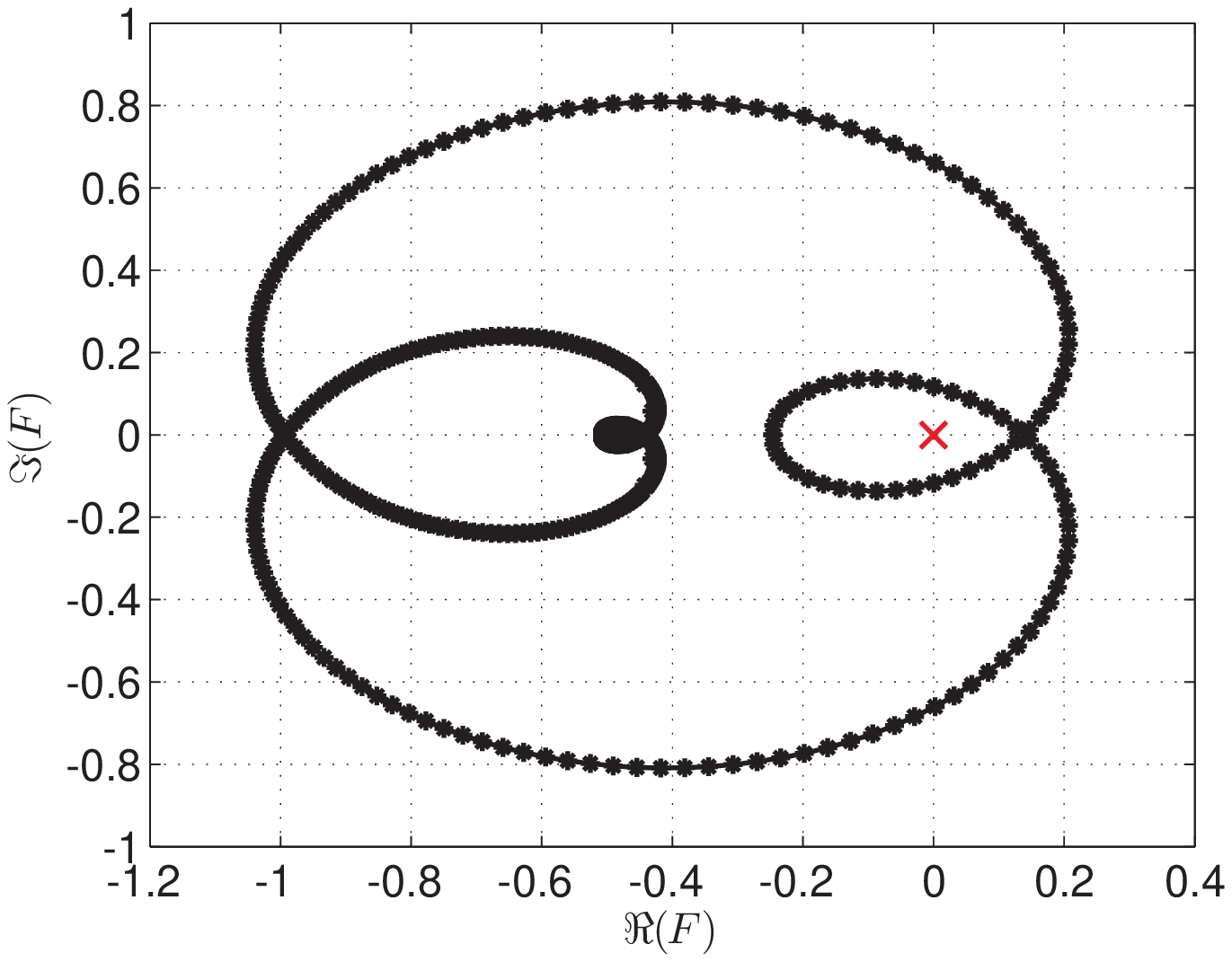}}
\subfloat[\label{fig:Nyquist-alpha-4.05-beta-0.005}$\alpha=4.05$, $\beta=0.005$.
Highly unstable case with twenty eigenvalues, ten shown in Fig. \ref{fig:Spectrum}
and their complex conjugates.]{\includegraphics[height=2.5in]{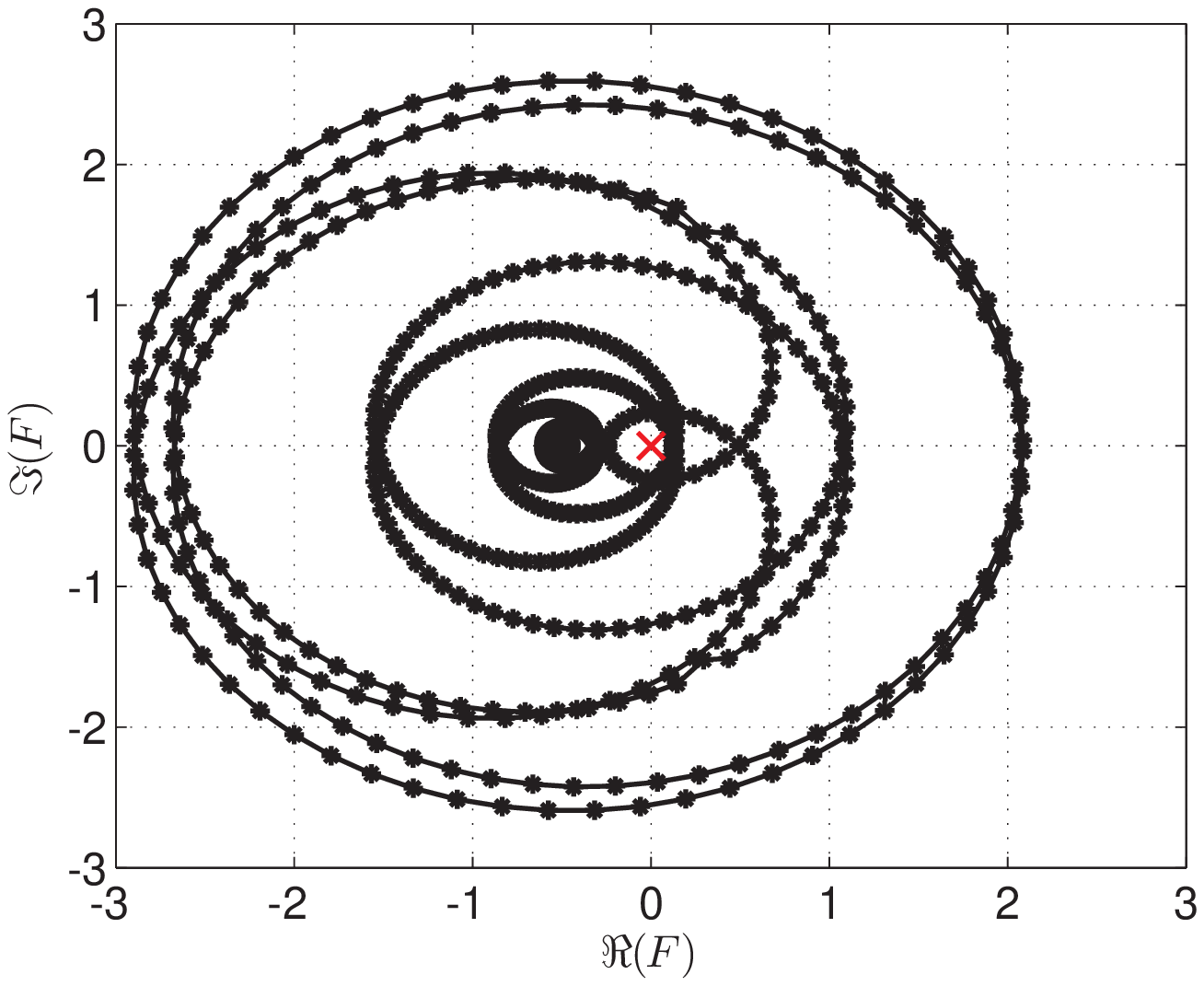}

}
\par\end{centering}

\caption{\label{fig:argument-principle}Values of $w=F\left(z\right)$ along
a large semi-circle in the right-half plane of the $z$-plane (radius
10 for \ref{fig:Nyquist-alpha-4.05-beta-0.05} and 100 for \ref{fig:Nyquist-alpha-4.05-beta-0.005}),
plotted in the $F$-plane. The total number of loops around the origin
in the $F$-plane gives the winding number, which is equal to the
number of unstable eigenvalues. }
\end{figure}

\subsubsection{The neutral curves}

We follow the first five unstable eigenvalues (ordered according to
their imaginary part) and show their neutral curves in Fig. \ref{fig:The-neutral-curve}.
We see that for large values of $\beta$, the lowest frequency eigenvalue
is the one that first becomes unstable, but for very small values
of $\beta$, the stability of the traveling wave is controlled by
the higher frequency perturbations. Moreover, the smaller the $\beta$,
the higher the frequency of the most unstable mode, consistent with
our earlier calculation of the square-wave-limit pathology. The whole
unstable region is given by the union of the unstable regions for
each eigenvalue and is generally located at large-enough $\alpha$
for any given $\beta$ or small-enough $\beta$ for any given $\alpha$. 

\begin{figure}[h]
\begin{centering}
\includegraphics[width=3in]{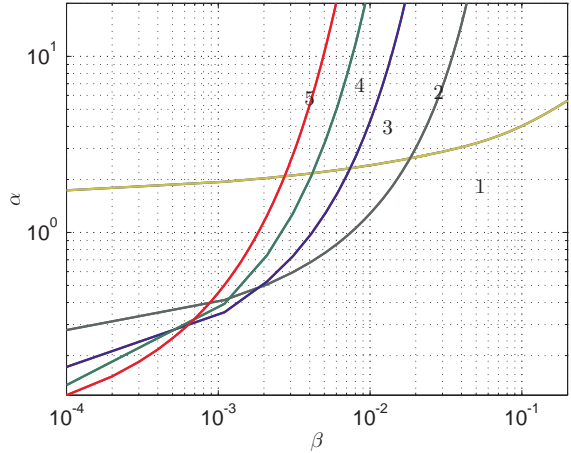}
\par\end{centering}

\caption{\label{fig:The-neutral-curve}The neutral curves for the first five
eigenvalues. The numbers next to each curve correspond to the index
of the eigenvalue. Below the envelope of the curves, we have discrete
spectral stability; in fact, numerical calculations indicate that
the solutions are stable at these parameters.}
\end{figure}

\subsection{Numerical simulations}

The previous section was concerned with the linear stability of traveling
wave solutions of (\ref{eq:KFRexample_dimensionless}). We were able
to compute the spectrum of unstable modes and obtain the neutral curves
in the $\alpha-\beta$ parameter space. In this section, we investigate
the behavior of solutions in the non-linear regime by numerically
solving the PDE using the WENO algorithm described in the Appendix.
All the simulations start with a steady-state solution, and instabilities
(when present) are triggered by the numerical discretization error
alone. The goal of this section is to demonstrate that, as in detonation
waves in the reactive Euler equations, the shock-dynamics goes through
the Hopf bifurcation followed by a period doubling cascade, when the
sensitivity parameter, $\alpha$, is varied, suggesting a possible
chaotic regime for large-enough $\alpha$.

\subsubsection{Linear growth and comparison with stability analysis }

We first compare the results obtained from the linear stability analysis
with the numerical results from the simulation. We perform a least-squares
fit on the deviation from the steady-state value of the form $\sum_{k=1}^{n}c_{k}e^{\sigma_{r_{k}}t}\cos(\sigma_{i_{k}}t+\delta_{k})$,
where $n$ is the number of unstable eigenvalues found in the linear
stability analysis. For instance, when $\beta=0.1$ and $\alpha=4.05$,
we expect from Fig. \ref{fig:Spectrum} one unstable mode to appear,
and thus, at least for a small time interval, we expect the solution
to behave like $e^{\sigma_{k}t}$, up to translation and scaling.
The results obtained from the comparison are presented in Table \ref{tab:Spectra-theory-vs-fit}.
We restrict ourselves to fitting up to two eigenvalues ($8$ parameters),
and fit up to a time when the perturbation is of the order $10^{-7}$.
The original perturbation is of the order $10^{-15}$. 

\begin{table}[h]
\begin{centering}
\begin{tabular}{|c|c|c|}
\hline 
$\beta$ & $\sigma$ from theory & $\sigma$ from numerics\tabularnewline
\hline 
\hline 
$0.10$ & $\begin{array}{cc}
 & 0.00309+0.38144i\end{array}$ & $\begin{array}{cc}
 & 0.00311+0.38152i\end{array}$\tabularnewline
\hline 
$0.01$ & $\begin{array}{cc}
 & 0.20092+0.30431i\\
 & 0.61295+3.78512i
\end{array}$ & $\begin{array}{cc}
 & 0.20581+0.29964i\\
 & 0.61298+3.78507i
\end{array}$\tabularnewline
\hline 
\end{tabular} 
\par\end{centering}

\caption{\label{tab:Spectra-theory-vs-fit}Comparison of eigenvalues from stability
analysis and from numerics at $\alpha=4.05$.}
\end{table}
The first case of $\beta=0.1$ in Table \ref{tab:Spectra-theory-vs-fit}
is near the neutral curve, and both the growth and frequency of the
perturbation are well captured by the linear stability predictions.
Simulations show that for this ``slightly unstable'' regime, the
predicted frequency is valid well into the nonlinear regime, an observation
often made in detonation simulations as well. In the second case,
when $\beta=0.01$, we see a larger discrepancy between the linear
theory and the numerical simulations, especially when capturing the
effect of the least unstable mode. This is to be expected, since the
effects of all unstable modes except for the most unstable one quickly
become negligible as the dominant mode starts to grow. This second
case is far from the neutral curve and much more unstable, with the
growth rate two orders of magnitude larger than in the first case.
Very fast growth of the perturbations is likely to result in nonlinear
effects starting to play an important role.

\subsubsection{Limit cycles and period-doubling bifurcations}

We now study the long-time asymptotic behavior of solutions that start
from a small perturbation (given by the discretization error) of the
initial steady-state solution. The shock value of the solution, $u_{s}(t)$,
is analyzed. For all the simulations that follow, we fix $\beta=0.1$
and vary $\alpha$. When $\alpha$ slightly exceeds the critical value
$\alpha_{c}\approx4.02$, predicted by the linear analysis as the
neutral boundary, the numerical solutions show that the steady-state
solution is unstable with the long-time evolution leading to a limit
cycle. 

For a range of $\alpha$ between $\alpha_{c}$ and $\alpha_{1}\approx4.72$,
the long time dynamics is that of a simple limit cycle (Fig. \ref{fig: us_periods}(a)).
Subsequent increase of $\alpha$ leads to a period doubling bifurcation.
When $\alpha$ is between $\alpha_{1}$ and $\alpha_{2}\approx4.91$,
we observe the limit cycle shown in Fig. \ref{fig: us_periods}(b).
This period doubling process continues until eventually, at $\alpha=\alpha_{\infty}\approx4.97$,
the solution (apparently) becomes chaotic. Figure \ref{fig: us_periods}(c)
illustrates the behavior of $u_{s}\left(t\right)$ for very large
values of $t$ (around $20,000$), when all the transients are likely
to have vanished. The respective power spectra, computed using a large
time window, $10,000<t<22,000$, are also shown. In the periodic case,
the power spectrum is clearly marked by peaks in the natural frequency
and its harmonics, as seen in Fig. \ref{fig: us_periods}(b,c). In
Fig. \ref{fig: us_periods}(c), we see that, although there is a dominant
frequency in the signal, many other frequencies are present, indicating
possible aperiodicity or chaos. Further analysis of the computational
results is required to establish whether the solution is indeed chaotic,
which is done in the subsequent sections.

\begin{figure}[h]
\subfloat[$\alpha=4.7$]{%
\begin{tabular}{c}
\includegraphics[height=4cm]{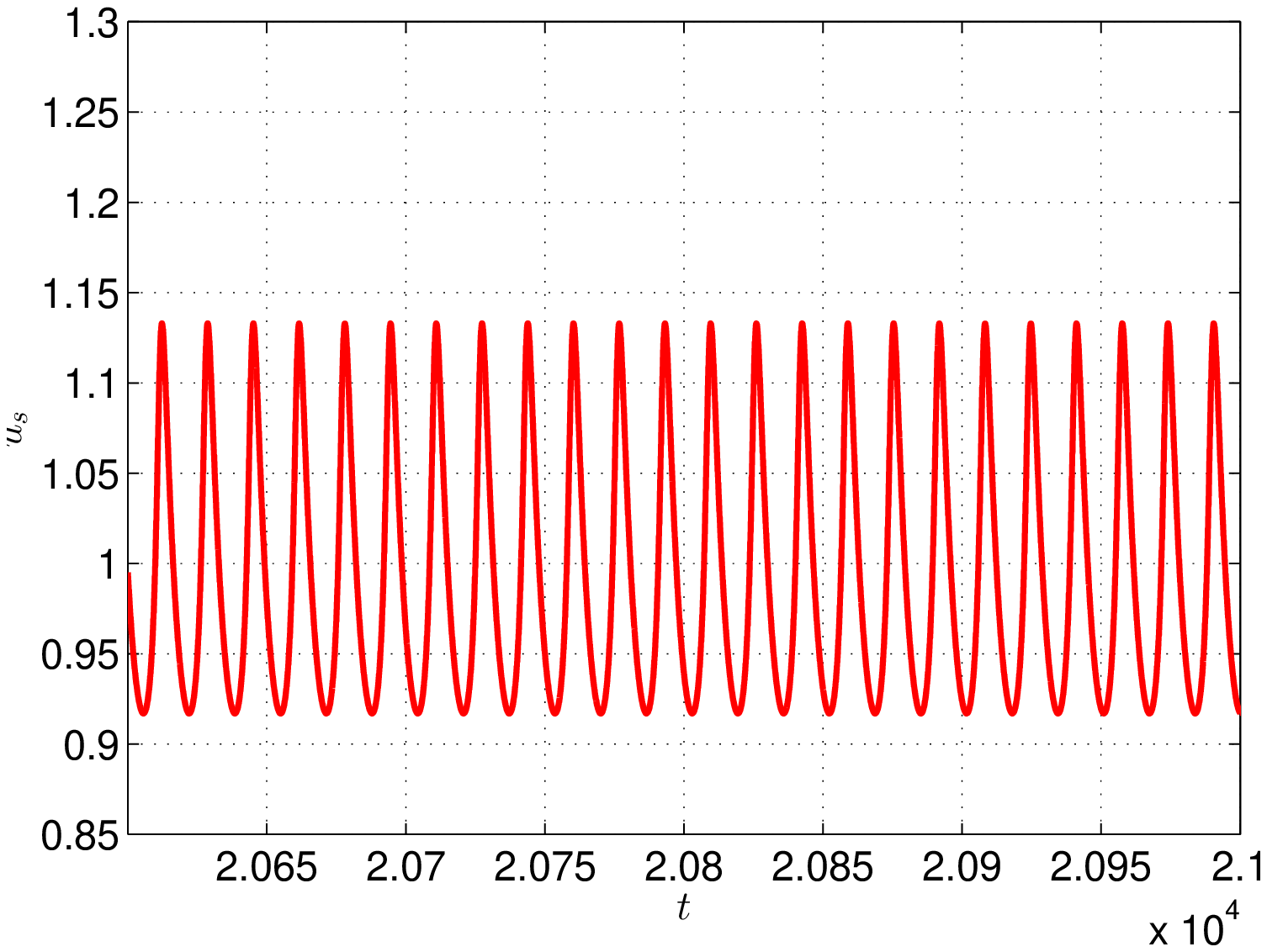}\tabularnewline
\includegraphics[height=4cm]{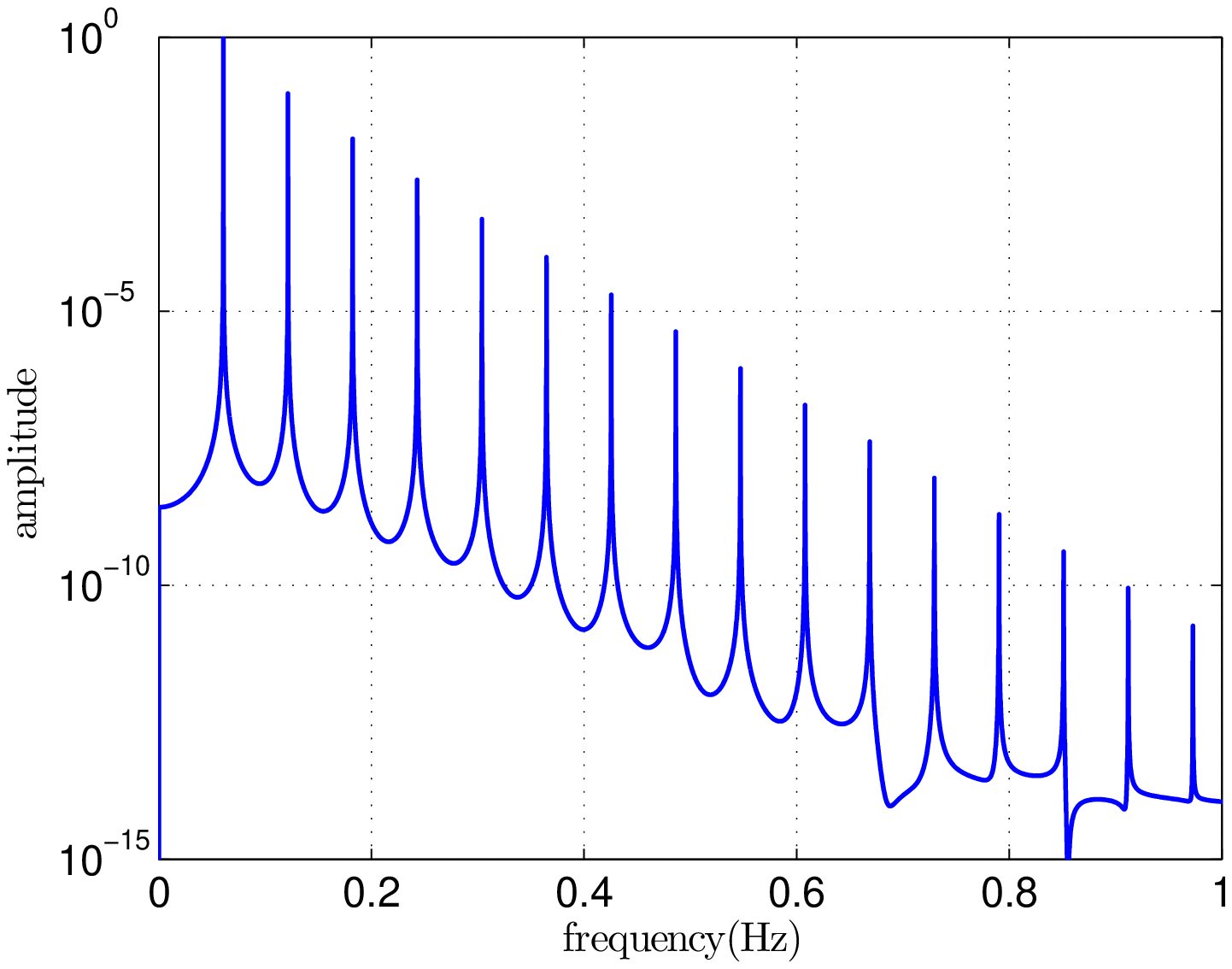}\tabularnewline
\end{tabular}}\subfloat[$\alpha=4.85$]{%
\begin{tabular}{c}
\includegraphics[height=4cm]{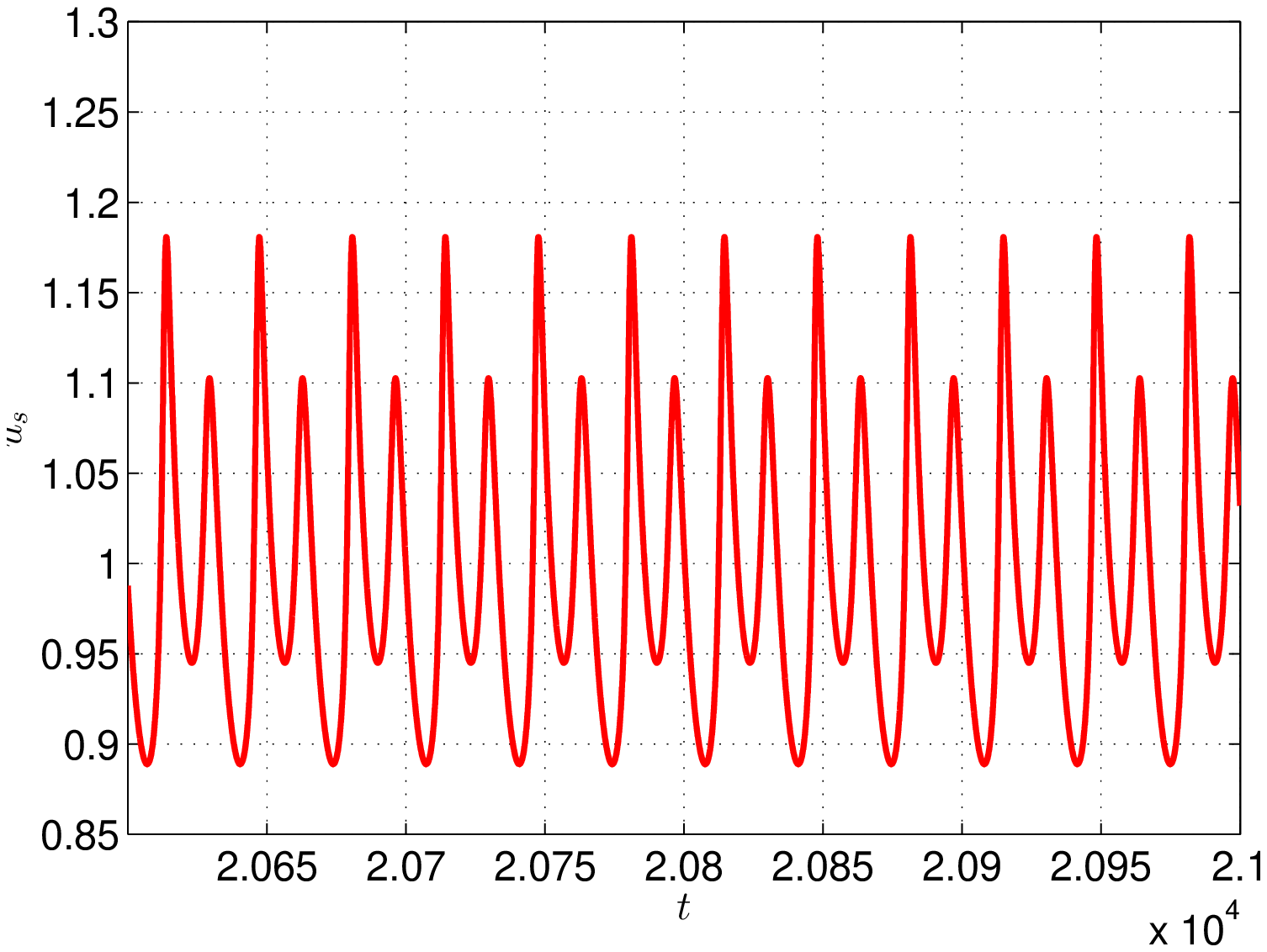}\tabularnewline
\includegraphics[height=4cm]{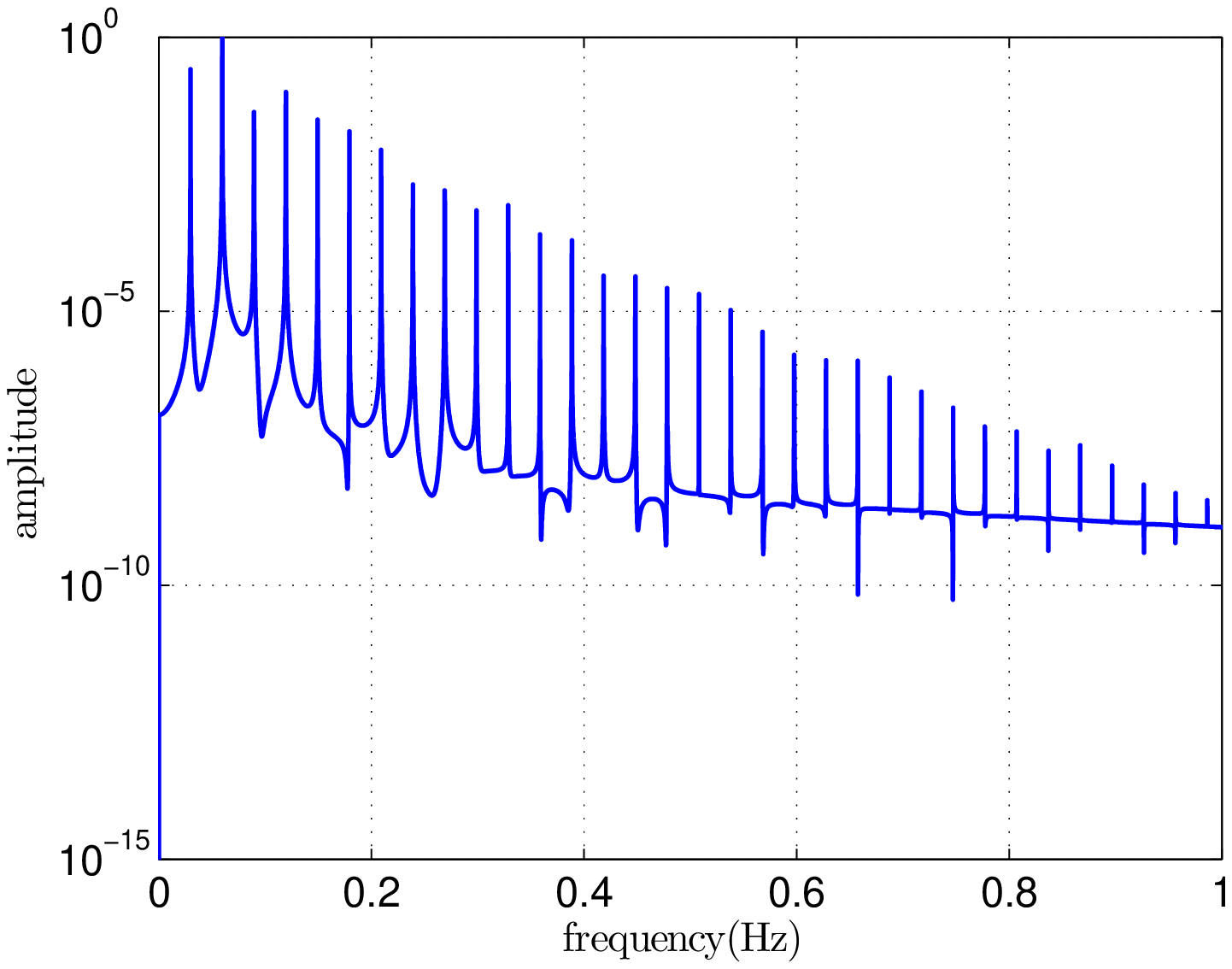}\tabularnewline
\end{tabular}}\subfloat[$\alpha=5.1$]{%
\begin{tabular}{c}
\includegraphics[height=4cm]{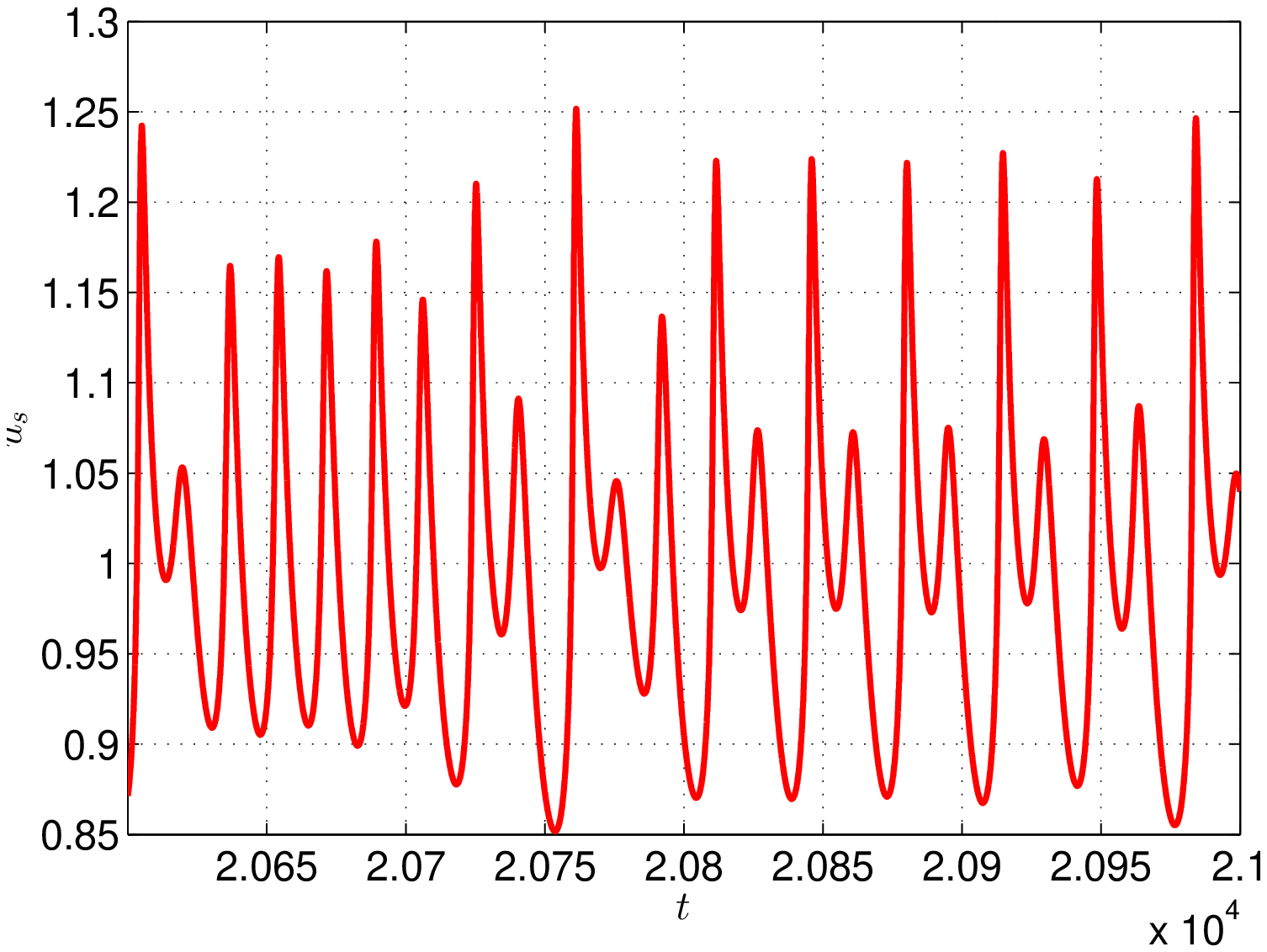}\tabularnewline
\includegraphics[height=4cm]{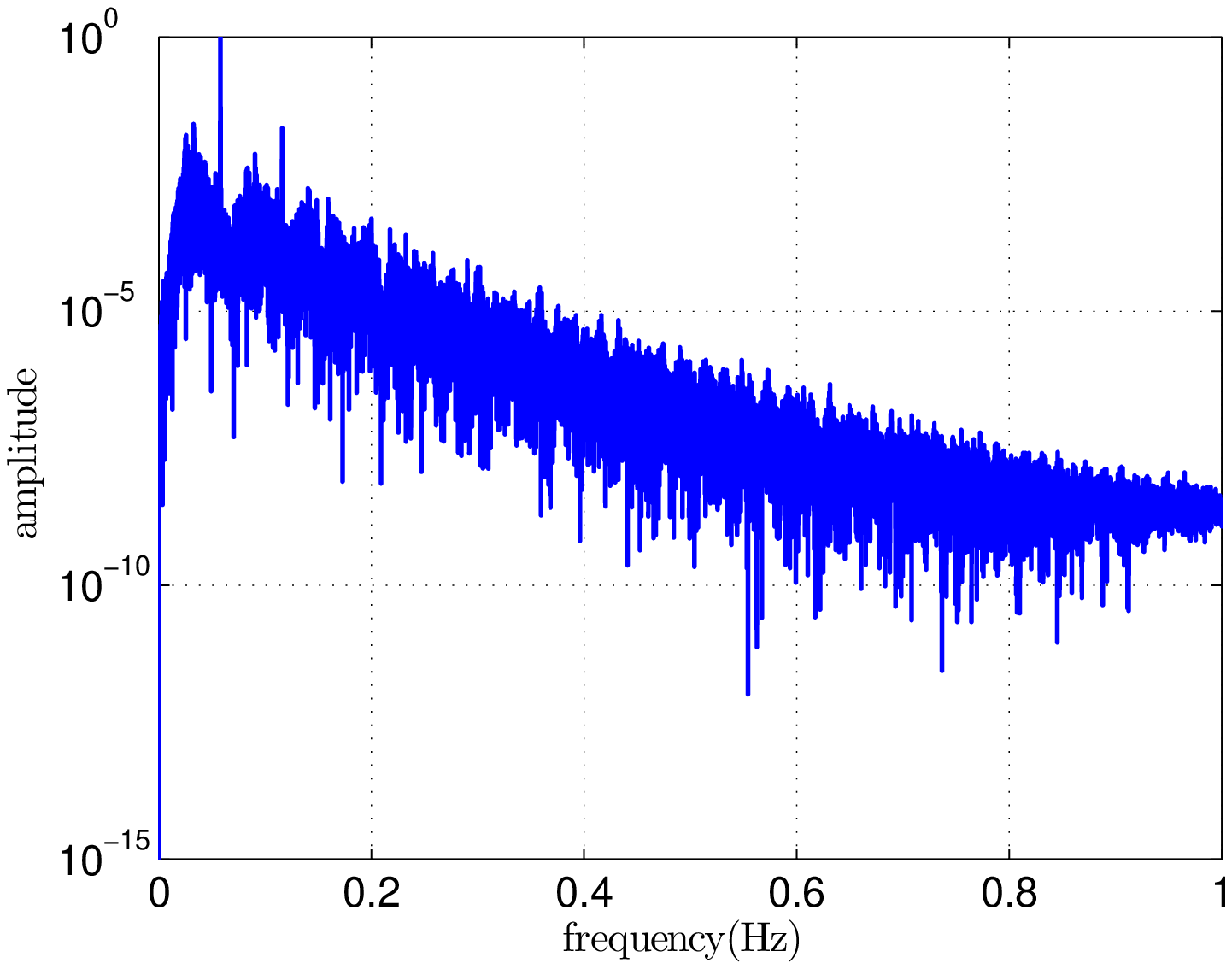}\tabularnewline
\end{tabular}}

\caption{\label{fig: us_periods}Top row -- $u_{s}\left(t\right)$ for $\beta$=0.1
and different values of $\alpha$. Bottom row -- corresponding power
spectra. }
\end{figure}

Although we focus on $u_{s}(t)$, the behavior presented in Fig. \ref{fig: us_periods}
is not unique to the shock value. That said, we must pick an ``interesting''
point, meaning a point close enough to the shock, if we want to capture
the rich dynamics. After the Hopf bifurcation occurs, $u\left(x,t\right)$
is periodic in time and as the bifurcation parameter ($\alpha$ in
this case) is increased further, $u\left(x,t\right)$ appears to become
chaotic. This is illustrated in Fig. \ref{fig:Characteristic-fields},
where the color represents $u\left(x,t\right)$ and the white lines
are the characteristics. 

\begin{figure}[h]
\begin{centering}
\subfloat[$\alpha=4.7$]{\includegraphics[clip,height=4cm]{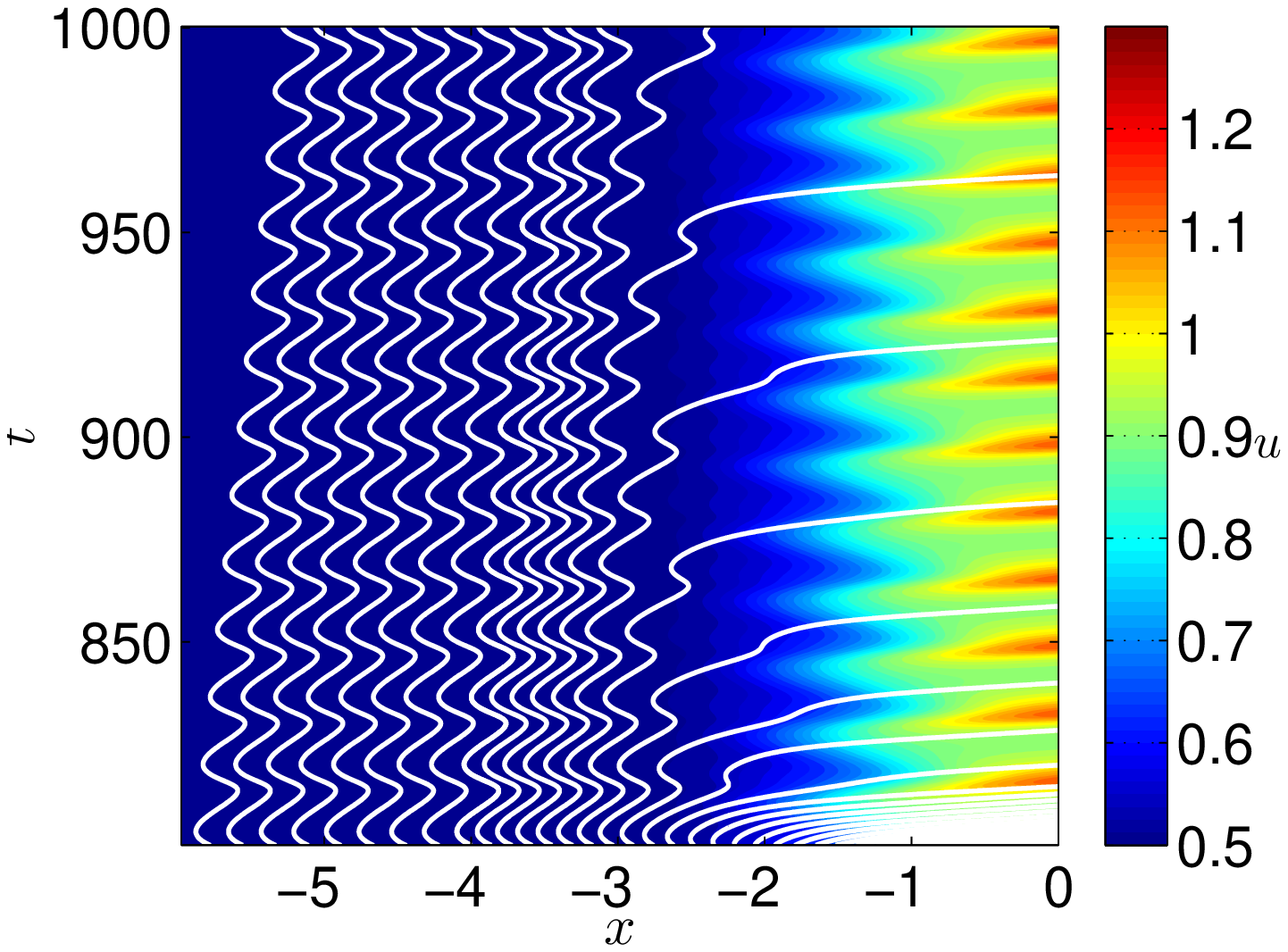}}\subfloat[$\alpha=4.85$]{\includegraphics[clip,height=4cm]{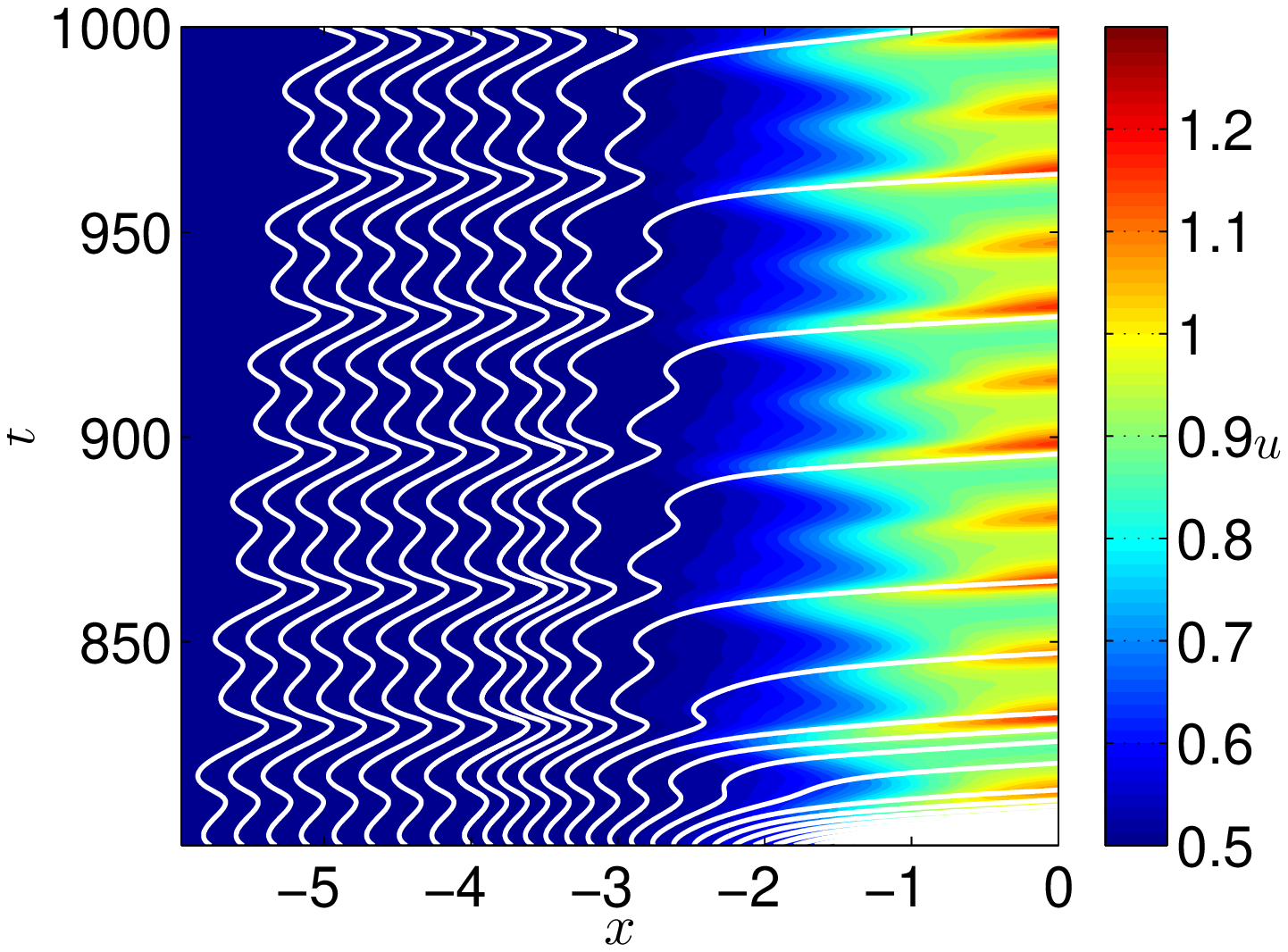}}\subfloat[$\alpha=5.1$]{\includegraphics[clip,height=4cm]{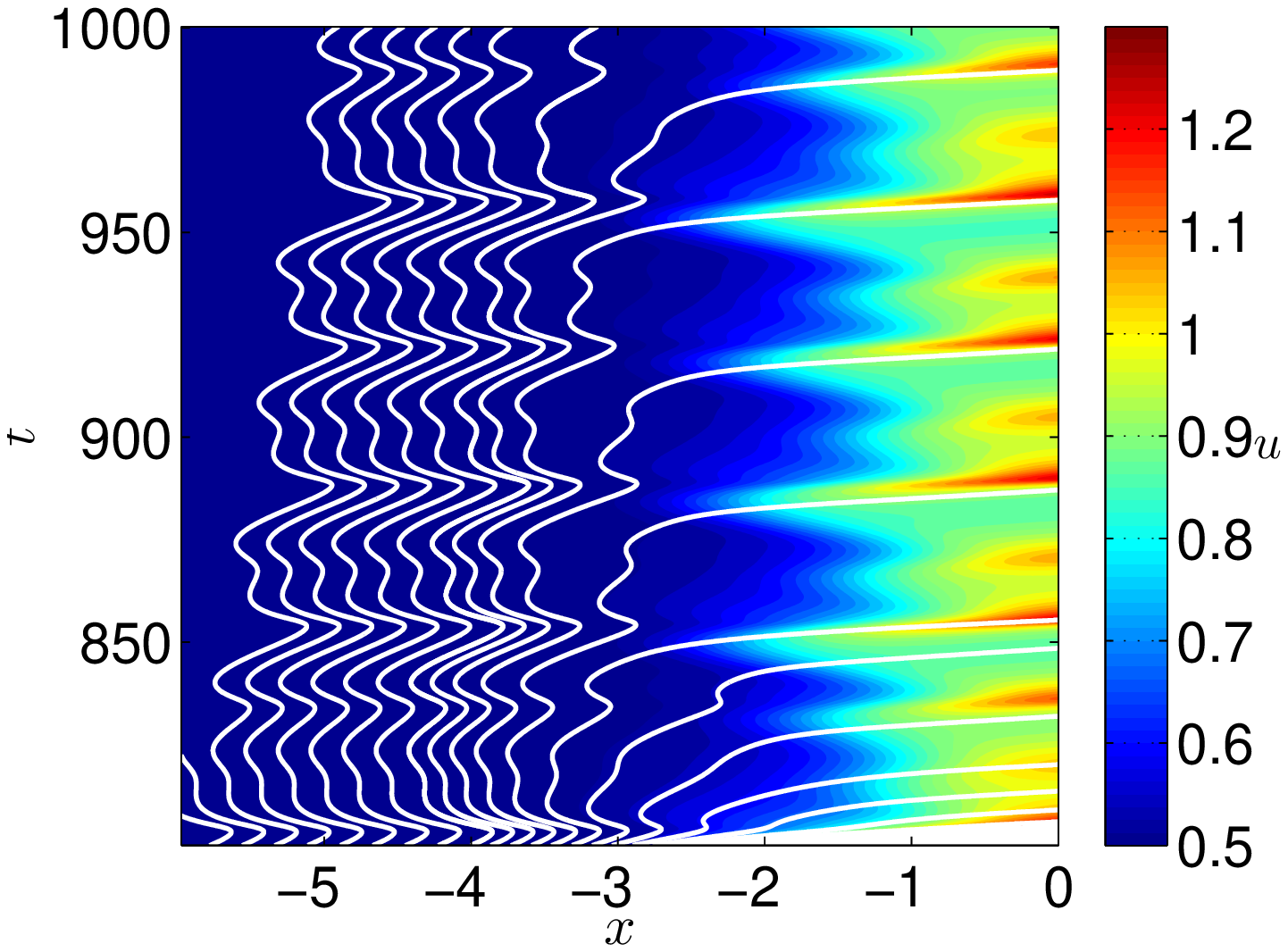}}
\par\end{centering}

\caption{\label{fig:Characteristic-fields}Characteristic fields (white curves)
at various $\alpha$, at periods 1, 2, and chaotic. The color shows
the magnitude of $u$.}
\end{figure}

The bifurcation process is best illustrated by means of a bifurcation
diagram, where the local maxima of the shock value, $u_{s}(t)$, are
plotted at different values of the bifurcation parameter $\alpha$
(Fig. \ref{fig:The-bifurcation-diagram}). 
\begin{figure}[h]
\begin{centering}
\includegraphics[width=5in]{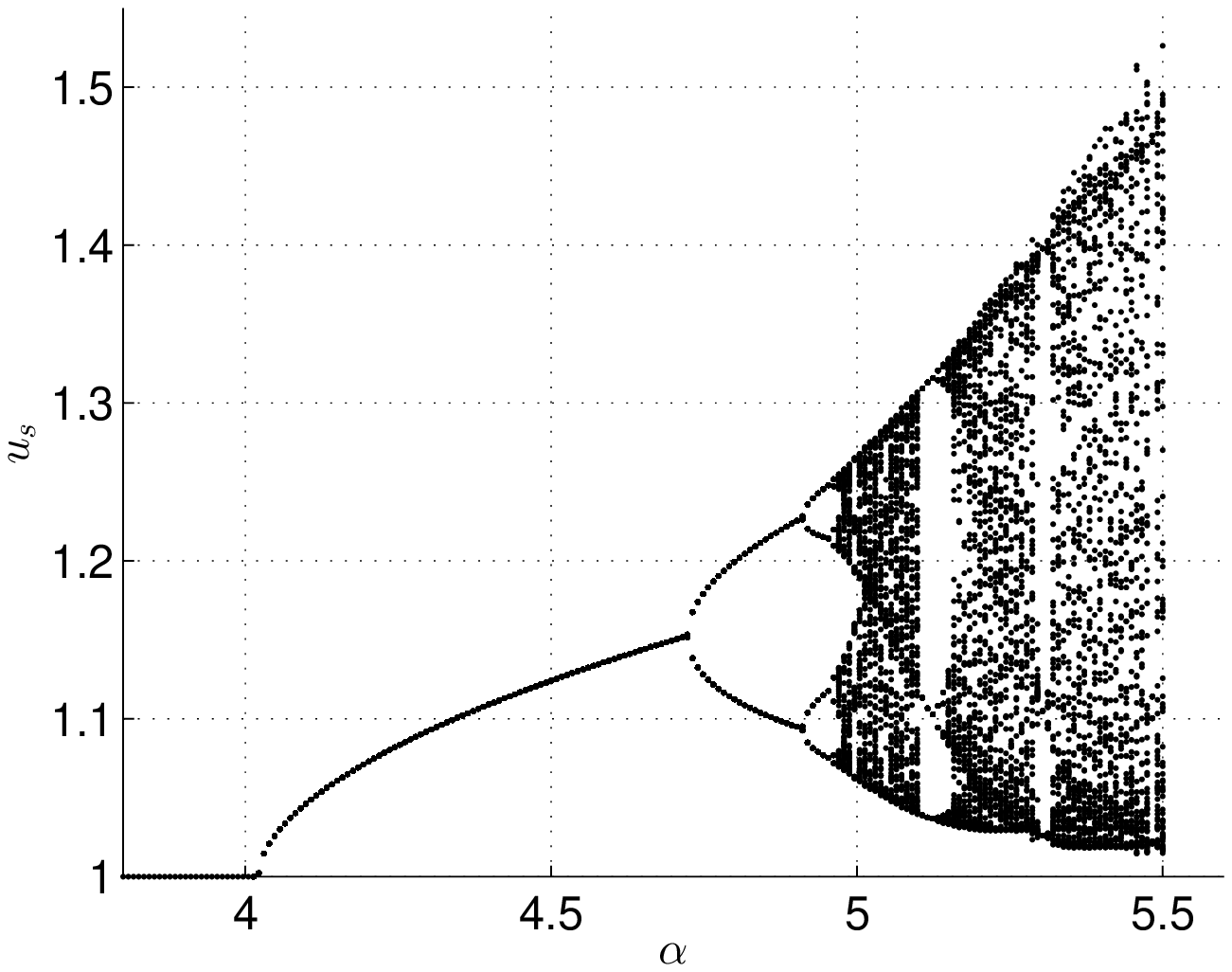}
\par\end{centering}

\caption{\label{fig:The-bifurcation-diagram}The bifurcation diagram at $\beta=0.1$.}
\end{figure}
The bifurcation points, presented in Table \ref{tab:Bifurcation-points},
are used to compute the Feigenbaum number, which appears to approach
the well-known constant $\delta\approx4.669$. 
\begin{table}[h]
\centering{}%
\begin{tabular}{|c|c|c|c|c|c|}
\hline 
$n$ & 1 & 2 & 3 & 4 & 5\tabularnewline
\hline 
\hline 
$\alpha_{n}$ & 4.02 & 4.7202 & 4.9100 & 4.95565 & 9.96553\tabularnewline
\hline 
$F_{n}$ & $\cdots$ & $\cdots$ & 3.68 & 4.15 & 4.62\tabularnewline
\hline 
\end{tabular}\caption{\label{tab:Bifurcation-points}Bifurcation points.}
\end{table}
 The bifurcation diagram in Fig. \ref{fig:The-bifurcation-diagram}
and the power spectra in Fig. \ref{fig: us_periods} all suggest (although
they do not prove) that the chaos in the system is real. In Section
\ref{sub:Time-series-analysis}, we analyze the apparently chaotic
series of $u_{s}\left(t\right)$ at very large $t$, i.e., on the
attractor. 
\begin{rem*}
An interesting feature of the example presented above is that, as
in the reactive Euler equations (e.g., \cite{kasimov2004dynamics}),
inner shocks can form inside the smooth region, $x<0$. These shocks
subsequently overtake the leading shock, rendering its dynamics non-smooth.
The inner-shock formation is simply due to the wave breaking and it
depends on the initial data as well as the parameters in $f$. For
example, as the parameter $\alpha$, which controls the shock-state
sensitivity, is increased, the characteristics are seen to converge
toward each other at large $t$, until, at a critical value of $\alpha$,
the characteristics collide into an inner shock. This shock then overtakes
the leading shock at $x=0$ as shown in Figure \ref{fig:inner-shock}.
A point to emphasize is that the characterization of chaos when such
non-smooth dynamics is present is not easy, in particular due to difficulties
of computing the solution with high accuracy. Our analysis of chaos
is therefore limited to moderate values of $\alpha$, when we know
that the internal shock does not form, yet a chaotic signal is observed.
\end{rem*}
\begin{figure}[h]
\begin{centering}
\includegraphics[width=5in]{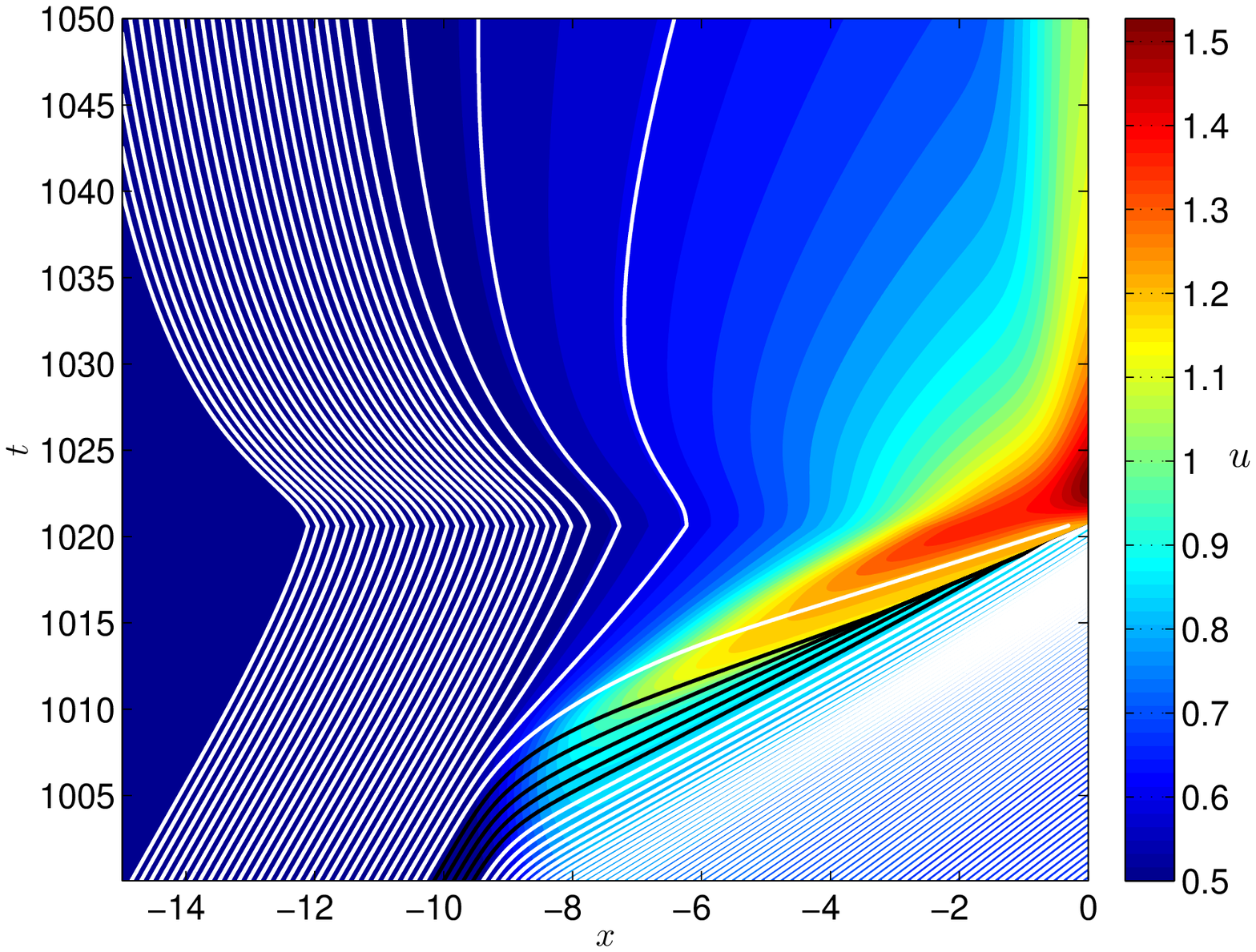}\caption{\label{fig:inner-shock}Formation of an internal shock wave. The color
shows the magnitude of $u$. The white curves are the forward characteristics.}

\par\end{centering}

\end{figure}

\subsection{\label{sub:Time-series-analysis}Time series analysis}

In this section, we use tools of nonlinear dynamical systems to understand
the shock signal. The shock signal represents a one-dimensional measurement
of the infinite dimensional phase space where the solutions live.
Relying on Takens' theorem \cite{takens1981detecting}, we embed the
signal in higher dimensions by choosing a delay, $\tau$, and an embedding
dimension, $m$ (note that choosing an appropriate $\tau$ is a delicate
question). We then use this embedded $m$-dimensional signal to compute
quantities of interest, such as the correlation dimension and the
largest Lyapunov exponent. The numerical calculations are performed
using the open source software OPENTSTOOL \cite{tstool}.

\subsubsection{Delay reconstruction of the attractor\label{sub:Delay-reconstruction-of-time-series}}

We embed the signal $u_{s}^{n}=u_{s}(t_{n})$ in $m$-dimensions by
creating the points 
\begin{eqnarray*}
p_{1} & = & (u_{s}^{1},u_{s}^{1+\tau},\cdots,u_{s}^{1+(m-1)\tau}),\\
p_{2} & = & (u_{s}^{2},u_{s}^{2+\tau},\cdots,u_{s}^{2+(m-1)\tau}),\\
\vdots & \vdots & \vdots\\
p_{N} & = & (u_{s}^{N},u_{s}^{N+\tau},\cdots,u_{s}^{N+(m-1)\tau}),
\end{eqnarray*}
where $N$ is limited by the number of available values of $u_{s}$.
The $m$-dimensional points $(p_{1},\cdots,p_{N})$ then live in an
attractor of dimension at most $m$. It was shown by Takens that provided
$m>2d+1$, where $d$ is the dimension of the attractor where $u_{s}$
lives, there exists a diffeomorphism between the reconstructed attractor
and the ``actual'' attractor (in the limit of the infinite amount
of noise-free data). This immediately allows us to use the reconstructed
attractor to compute quantities such as the correlation dimension
and the Lyapunov spectrum.

Notice that although in theory any choice of $\tau$ will allow such
reconstruction, in practice the situation is quite delicate. The finite
amount of noise-polluted data makes the choice of $\tau$ a non-trivial
issue, still subject of much current research. Since no fail-proof
method appears to exist, we choose $\tau$ as the first minimum of
the mutual information function of $u_{s}$. The reasons for such
a choice can be found in \cite{fraser1986independent}. In the next
subsection, we explore how the reconstructed attractor, its dimension,
and the largest Lyapunov exponent change as we vary the sensitivity
parameter, $\alpha$. We choose $\alpha=4.7,4.85,4.96,4.97,5,5.1$
and see how these quantities change as the dynamic goes from periodic
to chaotic.

\subsubsection{Largest Lyapunov exponent (LLE)}

A chaotic system is characterized by at least one positive Lyapunov
exponent. This means that information must be lost in the system as
time progresses. Predictability is thus highly limited. Because the
largest Lyapunov exponent determines the dominant rate at which information
is lost, we are primarily interested in the LLE . Several methods
are available to compute the LLE, and we choose to use the one presented
in \cite{rosenstein1993practical}. The algorithm used here is discussed
in the Appendix.

The sequence of period doubling observed in Fig. \ref{fig:The-bifurcation-diagram}
and Table \ref{tab:Bifurcation-points} suggests that the sequence
first saturates at $\alpha_{c}\approx4.97$. After this critical value,
the solution seems to become aperiodic, as indicated by its power
spectrum. We compute the LLE for values of $\alpha$ slightly below
and slightly above $\alpha_{c}$ in order to illustrate the drastic
change in the magnitude of LLE. The values of LLE are presented in
Table \ref{tab:LLE-Cdim}, where the error estimates are merely educated
guesses of a confidence interval obtained from running the algorithm
for different embedding dimensions (from dimension 3 to 10). It is
particularly difficult to obtain quantitative error estimates because
the sources of error are unknown and the algorithm requires some subjective
choice of a ``range'' (see the Appendix)

\begin{table}[h]
\begin{tabular}{|c|c|c|c|c|c|}
\hline 
 & $4.85$ & $4.96$ & $4.97$ & $5$ & $5.1$\tabularnewline
\hline 
\hline 
LLE & $0$ & $0$ & $0.0042\pm2\cdot10^{-4}$ & $0.01816\pm3\cdot10^{-5}$ & $0.0315\pm8\cdot10^{-4}$\tabularnewline
\hline 
$D_{C}$ & $1.0006\pm3\cdot10^{-4}$ & $1.002\pm2\cdot10^{-2}$ & $1.67\pm7\cdot10^{-2}$ & $1.87\pm3\cdot10^{-2}$ & $1.91\pm2\cdot10^{-2}$\tabularnewline
\hline 
\end{tabular} \caption{\label{tab:LLE-Cdim}The largest Lyapunov exponent and correlation
dimension for different values of $\alpha$, the bifurcation parameter.}
\end{table}

A study of the dependence of the LLE on the embedding dimension is
presented in the Appendix. Although precise error estimates are not
available, there is still some value in the predictions made; namely,
a clear difference is observed between $\alpha=4.96$ and $\alpha=4.97$,
which corresponds to the apparent saturation point of the bifurcation
diagram presented in Fig. \ref{fig:The-bifurcation-diagram}.

\subsubsection{Correlation dimension estimate }

While the Lyapunov exponent measures the rate at which information
is lost in a dynamical system, the correlation dimension gives an
upper bound on the number of degrees of freedom a system has. This
is an important concept to distinguish deterministic chaos from stochastic
chaos.\textbf{ }For simple attractors, the correlation dimension is
an integer, but for strange or chaotic attractors the dimension is
fractal. We compute the correlation dimension of our time series using
the algorithm presented in\textbf{ \cite{grassberger1983measuring}}.
The results for different values of $\alpha$ are shown in Table \ref{tab:LLE-Cdim}.

\section{Conclusions}

A simple model equation consisting of an inviscid Burgers' equation
forced with a term that depends on the current shock speed is analyzed
by calculating its steady-state solutions, the linear stability properties
of these solutions, and the non-linear, time-dependent evolution that
starts with the steady state as an initial condition. It is found
that the theory and numerical results for the model equation parallel
those of the reactive Euler equations of one-dimensional gas dynamics,
which have extensively been used to describe detonation waves. 

The steady-state theory of the model is analogous to that of the ZND
theory of detonation, describing both self-sustained and overdriven
solutions. The normal-mode linear stability theory of the model is
qualitatively similar to the detonation stability theory, reproducing
comparably complex spectral behavior. The nonlinear dynamics, computed
with a high-accuracy numerical solver, exhibit the Hopf bifurcation
from a stable solution to a limit cycle, together with a subsequent
cascade of period doubling bifurcations, resulting eventually in,
what is very likely, chaos. All of these features have their counterparts
in the solutions of the reactive Euler equations. The qualitative
agreement between the two systems, so drastically different in their
complexity, hints at a possibility that a theory for the observed
complex dynamics of detonations may be rather simple.\\

\section*{Aknowledgements}

The work by R.R.R. was partially supported by NSF grants DMS-1115278,
DMS-1007967, and DMS-0907955.

\bibliographystyle{plain}
\bibliography{/Users/aslankasimov/Dropbox/Biblioteka/akasimov-refs}

\appendix

\section*{\label{sec:Numerical-Method}Appendix. Numerical Algorithms}

\subsection*{PDE solver}

The hyperbolic system presented in this paper is solved using a method
of lines approach, in which we discretize in space and then evolve
the resultant ODE system in time. For the spatial discretization,
we use a five-point Weighted Essentially Non-Oscillatory (WENO) method
\cite{Shu-1998}. Our stencils are biased to the right by one point.
As usually done in WENO methods, we introduce a small parameter, $\epsilon$,
to guarantee that the denominators in the smoothness indicators of
the method do not become zero when calculating the weight coefficients.
For the problems investigated here, we experimented with $\epsilon$
between $10^{-5}$ and $10^{-10}$, and the solutions appear to be
unaffected by this choice. The chosen $\epsilon$ for all computations
was $\epsilon=10^{-6}$.

To avoid spurious oscillations, we use third-order Total Variation
Diminishing (TVD) Runge-Kutta time stepping algorithm \cite{gottlieb1998total}.
Convergence tests were performed using the steady-state solution in
the stable regime, for which fifth-order convergence in space was
obtained.

\subsection*{Largest Lyapunov exponent (LLE) }

The algorithm for the LLE consists of the following steps:
\begin{enumerate}
\item Given a time series $u_{s}^{n}$, embed it in an $m$-dimensional
space with delay $\tau$, as outlined in Section \ref{sub:Delay-reconstruction-of-time-series}.
\item For a given point $p_{i}$, find the closest point $p_{j_{i}}$ such
that $|i-j_{i}|>the\,\, mean\,\, period$, where the mean period is
estimated by the inverse of the dominant frequency of the power spectrum.\textbf{ }
\item Define $d_{i}^{m}(n)=\|p_{i+n}-p_{j+n}\|$. Then, $d_{i}(n)$ represents
the divergence between trajectories starting at $p_{i}$ and $p_{j_{i}}$. 
\item Choose $N$ points randomly on the attractor and compute an average
divergence of trajectories by $d^{m}(n)=\frac{1}{N}\sum_{l=1}^{N}d_{l}^{m}(n)$.
The number $N$ is limited either by the amount of available data
or by computational restrictions.
\item Plot $\log(d^{m}(n))$ versus $n\Delta t$. 
\item Repeat steps 1-5 for different values of embedding dimension, $m$,
and find a region $t_{min}<t<t_{max}$ such that the plot of $\log(d^{m}(n))$
vs. $n\Delta t$ is nearly a straight line for the values of $m$
used. 
\item Do a least squares fit in the region $t_{\min}<t<t_{\max}$ to extract
$\lambda_{1}^{m}$ for each embedding dimension $m$. 
\item If the values of $\lambda_{1}^{m}$ do not vary much for a wide range
of embedding dimensions, $m$, let $\lambda_{1}$ be the average over
all embedding dimensions computed. 
\end{enumerate}
The algorithm suggested above, which is presented in \cite{rosenstein1993practical},
has some parameters that are not objectively chosen. The value of
$\lambda_{1}$ depends, among other things, on the choices of $\tau$,
the range of $m$ considered, the choices of $t_{\min}$ and $t_{\max}$,
and on $N$. Of course, it also depends on the quality of the data
set and the amount of noise present in it. In \cite{rosenstein1993practical},
a numerical study of this parameter dependence is performed, and it
is claimed that the algorithm is rather robust. In our study, we use
the range $3\leq m\leq20$, fix $\tau=150$, choose $N=20,000$, and
choose $t_{\min}$ and $t_{\max}$ by looking at the plot of $\log(d)\text{ vs }t$.
A typical plot is show in Fig. \ref{fig:Divergence-of-Trajectories-and-Correlation-dimension},
where $\alpha=5$, $t_{\min}=100$, and $t_{\max}=200$. 

\begin{figure}[h]
\subfloat[The divergence of trajectories for different embedding dimensions.]{\includegraphics[height=2.5in]{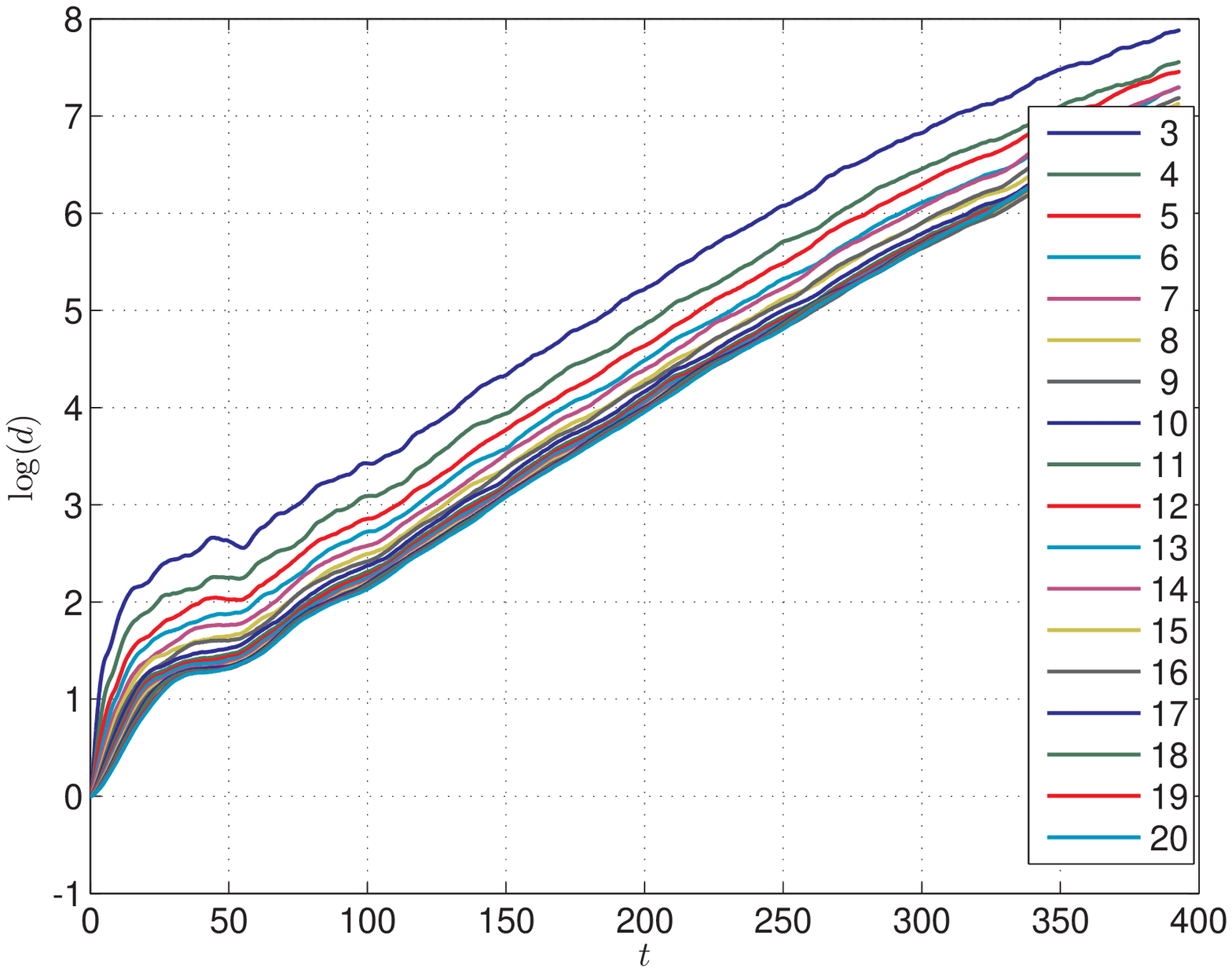}

}\subfloat[Plots of $\log(C_{r})$ vs $\log(r)$ for different embedding dimensions.]{\includegraphics[height=2.5in]{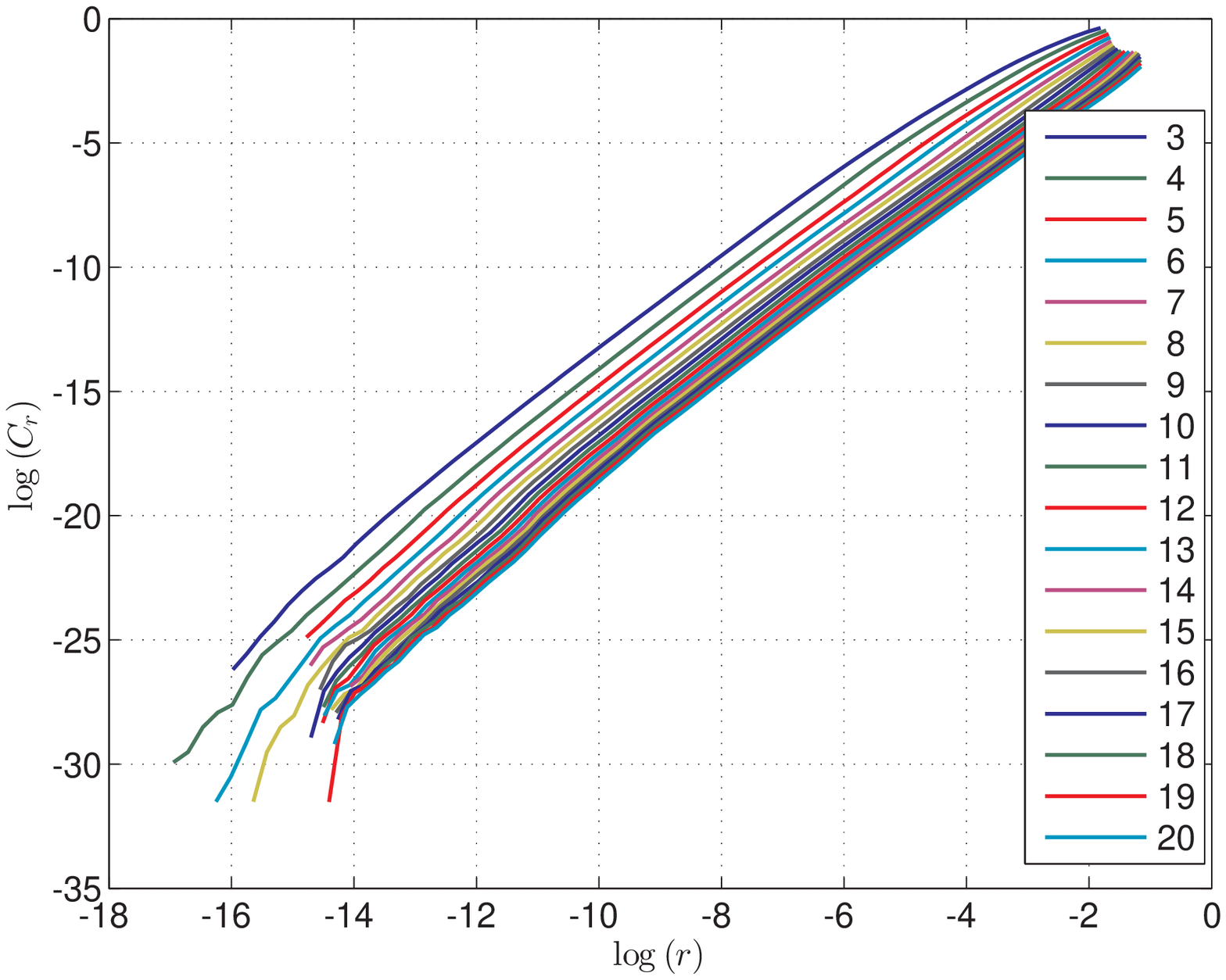}}\\
\subfloat[The largest Lyapunov exponent estimated from (A).]{\includegraphics[height=2.5in]{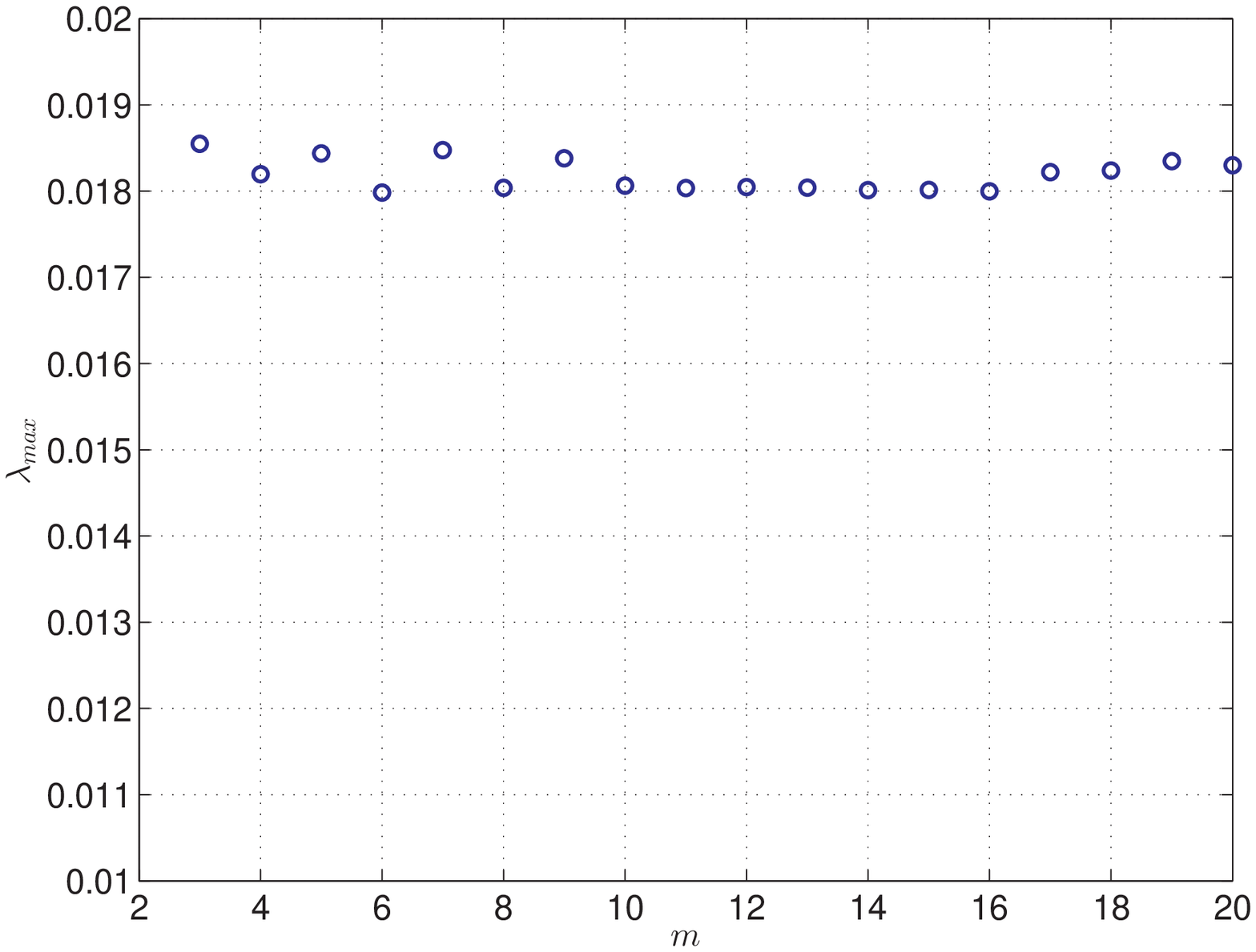}

}\subfloat[The correlation dimension estimated from (B).]{\includegraphics[height=2.5in]{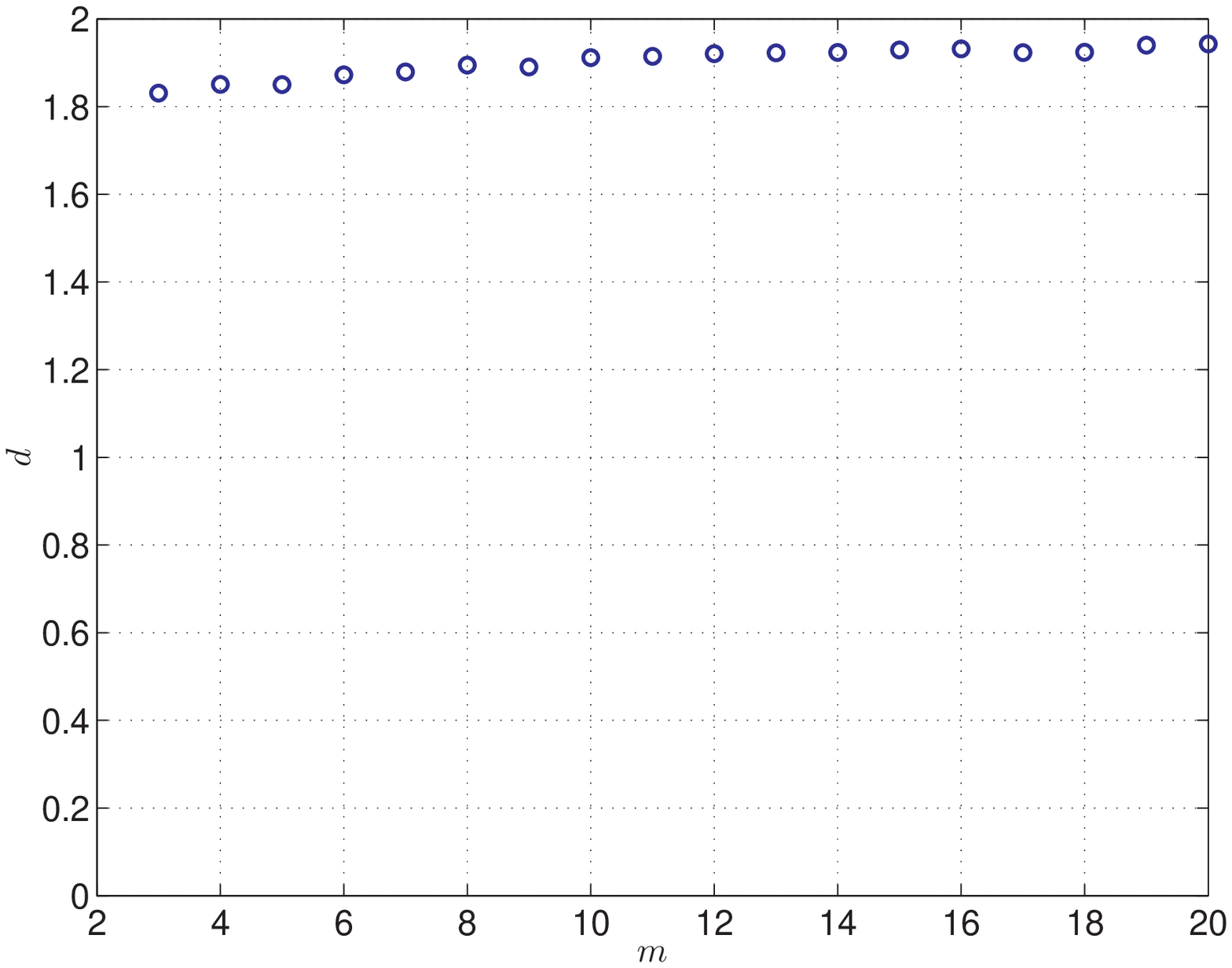}

}

\caption{\label{fig:Divergence-of-Trajectories-and-Correlation-dimension}Dependence
of the largest Lyapunov exponent and correlation dimension on the
choice of the embedding dimension.}

\end{figure}
In Table \ref{tab:LLE-different-dimensions}, we present the values
of LLE calculated for each given dimension from Fig. \ref{fig:Divergence-of-Trajectories-and-Correlation-dimension}. 

\begin{table}[h]

\begin{tabular}{|c|c|c|c|c|c|c|c|}
\hline 
m & 3 & 4 & 5 & 8 & 10 & 15 & 20\tabularnewline
\hline 
\hline 
LLE & 0.0185  & 0.0182 & 0.0184 & 0.0180 & 0.0181 & 0.0180 & 0.0183\tabularnewline
\hline 
\end{tabular}\caption{\label{tab:LLE-different-dimensions}The LLE for different embedding
dimensions.}

\end{table}

For other values of $\alpha$, the same methodology is applied. The
values shown in Table \ref{tab:LLE-Cdim} are obtained by averaging
the LLE over multiple dimensions. The error estimates are the maximum
differences between the averages and the entries.

\subsection*{Correlation dimension}

The algorithm for computing the correlation dimension follows that
of \cite{grassberger1983measuring}. It consists of the following
steps:
\begin{enumerate}
\item Given a time series $u_{s}^{n}$, embed it in an $m$ dimensional
space with delay $\tau$, as outlined in Section \ref{sub:Delay-reconstruction-of-time-series}.
\item Construct a grid $\bar{r}=(r_{1},\cdots,r_{L})$ where $r_{1}>\min_{i,j}(\|u_{s}^{i}-u_{s}^{j}\|)$
and $r_{L}<\max_{i,j}(\|u_{s}^{i}-u_{s}^{j}\|)$. 
\item For each $r_{i}$ define the correlation sum, at a given dimension
$m$, to be $C^{m}(r_{k})=\frac{1}{N^{2}}\sum_{i,j=1}^{N}\theta(r_{k}-\|u_{s}^{i}-u_{s}^{j}\|)$.
\item Plot $\log(C^{m}(r_{k}))$ versus $\log(r_{k})$.
\item Repeat steps 1-4 for different values of the embedding dimension $m$,
and find a region $r_{min}<r<r_{max}$ such that the plot of $\log(C^{m}(r))$
vs. $\log(r)$ is nearly a straight line for the values of $m$ used. 
\item Do a least squares fit over the region $r_{\min}<r<r_{\max}$ to extract
$D_{C}^{m}$ for each embedding dimension $m$. 
\item If the values of $D_{C}^{m}$ do not vary much for a wide range of
embedding dimensions $m$, let $D_{C}$ be the average over all embedding
dimensions.
\end{enumerate}
Similar to the LLE calculation, the computed value of $D_{C}$ depends
on many parameters that cannot be objectively chosen. The choices
of $\tau,\ m,\ r_{\min},\ r_{\max}$ in particular have an appreciable
effect on the value of $D_{C}$. In our study, we use the range $3\leq m\leq20$,
fix $\tau=150$, choose $N=5000$, and choose $r_{\min}$ and $r_{\max}$
by looking at the plot of $\log(C^{m}(r))\text{ vs }\log(r)$. A typical
plot is shown in Fig. \ref{fig:Divergence-of-Trajectories-and-Correlation-dimension}(a),
where $\alpha=5$, $\log(r_{\min})=-8$, and $\log\left(r_{\max}\right)=-4$.

In Table \ref{tab:Correlation-dimension-for-different-dimensions},
we show the computed values of the correlation dimension for the data
presented in Fig. \ref{fig:Divergence-of-Trajectories-and-Correlation-dimension}(b).
Notice that the variability here is much higher than in the computation
for the largest Lyapunov exponent. 

\begin{table}[h]
\begin{tabular}{|c|c|c|c|c|c|c|c|}
\hline 
$m$ & 3 & 4 & 5 & 8 & 10 & 15 & 20\tabularnewline
\hline 
\hline 
$D_{c}$ & 1.8306 & 1.8507 & 1.8500 & 1.8904 & 1.9147  & 1.9317 & 1.9432\tabularnewline
\hline 
\end{tabular}\caption{The correlation dimension for different embedding dimensions. \label{tab:Correlation-dimension-for-different-dimensions}}
\end{table}

\end{document}